\documentclass[letterpaper, 11pt]{article}

\def\conf{0}

\usepackage[margin=1in]{geometry}
\usepackage{amsmath}
\usepackage{amssymb}
\usepackage{array}
\usepackage{amsthm}
\usepackage[colorlinks,urlcolor=blue,citecolor=blue,linkcolor=blue]{hyperref}
\usepackage{paralist}
\usepackage[usenames, dvipsnames]{color}
\usepackage{xfrac}
\usepackage{paralist}
\usepackage{color}
\usepackage{stackrel}
\usepackage{caption}
\usepackage	{booktabs}

\newcommand\numberthis{\addtocounter{equation}{1}\tag{\theequation}}

\newtheoremstyle{mythm}
{5pt}
{5pt}
{\itshape}
{}
{\bfseries}
{.}
{ }
{}
\theoremstyle{mythm}

\renewenvironment{proof}[1][\proofname]{%
	\par\pushQED{\qed}\normalfont%
	\trivlist\item[\hskip\labelsep\bfseries#1.]%
	\ignorespaces
}{%
	\popQED
}

\let\oldproofname=\proofname
\renewcommand{\proofname}{\rm\bf{\oldproofname}}

\newtheorem{thm}{Theorem}
\newtheorem{cor}[thm]{Corollary}

\newtheorem{lemma}{Lemma}[section]
\newtheorem{claim}[lemma]{Claim}
\newtheorem{corollary}[lemma]{Corollary}
\newtheorem{fact}[lemma]{Fact}

\newtheorem{definition}{Definition}[section]

\newcommand{\mnote}[1]{{\color{red}$\star$}	\marginpar{\tiny\bf
		\begin{minipage}[t]{0.5in}
			\raggedright #1
\end{minipage}}}
\newcommand{\link}[1]{\hyperref[#1]{\color{black}\textsf{#1}}}

\newcommand{\eps}{\varepsilon}

\renewcommand{\th}{^{\textrm{th}}}

\DeclareMathOperator{\poly}{poly}

\newcommand{\mA}{\mathcal{A}}
\newcommand{\mAlg}{\mathcal{ALG}}
\newcommand{\mG}{\mathcal{G}}
\newcommand{\mFA}{\mathcal{F}^{\mA}}
\newcommand{\mN}{\mathcal{N}}

\newcommand{\mC}{{\mathcal{	C}}}
\newcommand{\mD}{{\mathcal{D}}}
\newcommand{\mOC}{{\mathcal{O}}}

\newcommand{\vT}{\vec{T}}
\newcommand{\vC}{\vec{C}}
\newcommand{\vI}{\vec{I}}

\newcommand{\mI}{\mathcal{I}}

\newcommand{\mR}{{\mathcal{R}}}
\newcommand{\mS}{\mathcal{S}}

 \newcommand{\mFAI}{\mathcal{F}^{\mA, \vI}}
\newcommand{\wtAI}{\omega^{\mA, \vI}}

\newcommand{\mPAI}{\mathcal{P}^{\mA,\vI}}

 \newcommand{\Approx}{\hyperref[alg:const-approx]{\color{black}{\bf Approx-Cliques}}}
\newcommand{\Main}{\hyperref[alg:main]{\color{black}{\bf Approx-Cliques-with-Is-Active}}}
\newcommand{\IsActive}{\hyperref[alg:is-active]{\color{black}{\bf Is-Active}}}
\newcommand{\SampleSet}{\hyperref[alg:sample-set]{\color{black}{\bf Sample-a-Set}}}

\newcommand{\hnk}{\widehat{n}_k}
\newcommand{\tnk}{\widetilde{n}_k}
\newcommand{\tm}{\widetilde{m}}

\def\withcolors{0}
\ifnum\withcolors=1
\newcommand{\tnote}[1]{{\color{red}{#1}}}
\newcommand{\dnote}[1]{{\color{blue}{#1}}}
\else
\newcommand{\tnote}[1]{{{#1}}}
\newcommand{\dnote}[1]{{{#1}}}
\renewcommand{\mnote}[1]{}
\fi

\newcommand{\tauu}{\tau^{\scriptscriptstyle  U}}
\newcommand{\taul}{\tau^{\scriptscriptstyle  L}}
\newcommand{\vtauu}{\vec{\tau}^{\scriptscriptstyle  U}}
\newcommand{\vtaul}{\vec{\tau}^{\scriptscriptstyle  L}}
\newcommand{\hck}{\widehat{c}_k}

\DeclareMathAlphabet{\mathpzc}{OT1}{pzc}{m}{it}
\newcommand{\EX}{{\rm Ex}}
\newcommand{\cldots}{\cdot \ldots\cdot}

\newcommand{\confeqn}[1]{
	\ifnum\conf=0
	\[#1\]
	\else
	$#1$
	\fi
}

\newcommand{\clk}{n_k}

\newcommand{\wt}{\omega}
\newcommand{\wtA}{\omega^{\mA}}
\newcommand{\wtP}{\wt^{\mA}}
\newcommand{\wtPI}{\wt^{\mA^*, \vI}}

\newcommand{\asIs}{\pi^{\mA^*}_{\vI}}

\newcommand{\pgood}{{$(\eps,\vec{\tau})$-good}}

\def\SumMindCliques{
\begin{proof}
		\ifnum\conf=1
	[Proof of Claim~\ref{clm:sum_min_d_cliques}]
	\fi
		An acyclic orientation of $G$ is obtained by directing every edge in $G$ such that
	the resulting digraph is acyclic. A standard fact regarding arboricity
	is the existence of an acyclic orientation $D$ of $G$ such that the outdegree of every vertex in $D$
	is at most $\alpha(G)$~\cite{MB83}. 
For any $t$-clique $T$ in $D$, let $\ell^D(T)$ be the least vertex in $T$, according to the ordering in $D$.
	For every $v \in V$, let $\mD_v = \{T: \ell^D(T) = v\}$. Since outdegrees in $D$ are bounded by $\alpha(G)$,
	$\forall v, |\mD_v| \leq \alpha(G)^{t-1}$. Furthermore, the sets $\mD_v$ form a partition of $\mC_t(G)$.
	\begin{equation}
	d(\mC_{t}(G)) = \sum_{T \in \mC_{t}(G)} d(T)  = \sum_{v \in V} \sum_{T \in \mD_v} d(T)
	\leq \sum_{v \in V} d(v) |\mD_v| \leq \alpha(G)^{t-1} \sum_{v \in V} d(v) = 2m\cdot\alpha(G)^{t-1}\;,
	\end{equation}
	and the claim is established.
\end{proof}
}

\def\BoundNkProof{
\begin{proof}
	\ifnum\conf=1
	[Proof of Claim~\ref{clm:d^k_leq_d^k-1}]
	\fi
	We rephrase the claim as: for all $t \geq 2$, for all graphs $H$,
	$n_t(H) \leq 2\alpha(H) n_{t-1}(H)/t$.
	We will prove the claim by induction on $t$. For the base case $t=2$, $n_2 = m(H) \leq n(H) \alpha(H)$.
	Assume that the claim holds for all values $i \leq t-1$.
	Fix an arbitrary graph $H$.
	For a vertex $v$, let $n_t(v)$ denote the number of $t$-cliques that $v$ participates in.
	Observe that
	$n_t(v) = n_{t-1}(H|_{\Gamma(v)})$ (recall that $H|_S$ is the subgraph induced by $S$).
	Note that the arboricity of a subgraph of $H$ is at most the arboricity of $H$.
	Therefore, by the induction hypothesis, for every vertex $v$,
	\begin{equation}
	n_t(v) = n_{t-1}(H|_{\Gamma(v)}) \leq \frac{2\alpha(H)}{t-1} \cdot  n_{t-2}(H|_{\Gamma(v)})=\frac{2\alpha(H)}{t-1} \cdot n_{t-1}(v)\;.
	\end{equation}
	Therefore,
	\begin{equation}
	n_t = \frac{1}{t}\sum_{v} n_t(v) \leq \frac{1}{t}\sum_{v}\frac{2\alpha(H)}{t-1} n_{t-1}(v) = \frac{2\alpha(H)}{t}\cdot n_{t-1} \;,
	\end{equation}
	and the claim is established.
\end{proof}
}

\begin{document}

\begin{titlepage}

	\title{Faster sublinear approximations of $k$-cliques \\ for low arboricity graphs}

	\author{
		Talya Eden
		\thanks{Tel Aviv University, \textit{ talyaa01@gmail.com}. This research was partially supported by a grant from the Blavatnik fund. The author is grateful to the Azrieli Foundation for the award of an Azrieli Fellowship.}
		\and  Dana Ron
		\thanks{Tel Aviv University, \textit {danaron@tau.ac.il}.
		This research was partially supported by the Israel Science Foundation grant No. 671/13 and by a grant from the Blavatnik fund.}
	\and  C. Seshadhri
		\thanks{University of California, Santa Cruz, \textit{ sesh@ucsc.edu}. This research was funded by NSF CCF-1740850 and NSF CCF-1813165.}
	}

	\date{}
	\maketitle
	\begin{abstract}
		Given query access to an undirected graph $G$, we consider the problem of
{computing a}
$(1\pm\eps)$-approximation
{of} the number of $k$-cliques in $G$. The standard query
model for general graphs allows for degree queries, neighbor queries, and pair queries.
Let $n$ be the number of vertices, $m$ be the number of edges, and $n_k$
be the number of $k$-cliques.
Previous work by Eden, Ron and Seshadhri (STOC 2018) gives an $O^*(\frac{n}{n^{1/k}_k} + \frac{m^{k/2}}{n_k})$-time algorithm
for this problem (we use $O^*(\cdot)$ to suppress $\poly(\log n, 1/\eps, k^k)$ dependencies).
Moreover, this bound is nearly optimal when the expression is sublinear in the size of the graph.

Our motivation is to circumvent this lower bound, by parameterizing the complexity in terms of \emph{graph arboricity}.
The arboricity of $G$ is a measure for the graph density ``everywhere''.
We design an algorithm for the class of graphs with arboricity at most $\alpha$, whose
running time is $O^*(\min\{\frac{n\alpha^{k-1}}{n_k},\; \frac{n}{n_k^{1/k}}+\frac{m \alpha^{k-2}}{n_k} \} )$.
We also prove a nearly matching lower bound.
For all graphs, the arboricity is $O(\sqrt m)$, so this bound subsumes all previous results on sublinear clique approximation.

As a special case of interest, consider minor-closed families of graphs, which have constant arboricity.
Our result implies that
for any minor-closed family of graphs,
there is a $(1\pm\eps)$-approximation algorithm for $n_k$ that has running time $O^*(\frac{n}{n_k})$.
Such a bound was not known even for the special (classic) case of
triangle counting in planar graphs.

	\end{abstract}

\thispagestyle{empty}

\end{titlepage}
	
	\tableofcontents
\setcounter{page}{1}
	\newpage

\section{Introduction}\label{sec:intro}

The problem of counting the number of $k$-cliques in a 
graph
is a fundamental problem in theoretical computer science, with a wide
variety of applications~\cite{HoLe70,ChNi85,Co88,portes2000social,EcMo02,milo2002network,Burt04,V09,becchetti2008efficient,foucault2010friend,BFN+14,SeKoPi11,JRT12,ELS13,Ts15,FFF15,JS17}.
This problem has seen a resurgence of interest because of its importance
in analyzing massive real-world graphs (like social networks and biological networks).
There are a number of clever algorithms for exactly counting $k$-cliques using matrix multiplications~\cite{NP85,EG04}
or combinatorial methods~\cite{V09}. However, the complexity of
these algorithms grows with $m^{\Theta(k)}$, where $m$ is the number of edges in the graph.

A line of recent work has considered this question from a sublinear approximation  perspective~\cite{ELRS, ERS18}.
Letting $n$ denote the number of vertices, and $n_k$ denote the number of $k$-cliques, the complexity is basically $O\left(\frac{n}{n_k^{1/k}}+ \frac{m^{k/2}}{n_k}\right)$ with a nearly matching lower bound. 

We study the problem of clique estimation in bounded arboricity graphs, with the hope
of circumventing the above lower bound.\footnote{The arboricity of a graph is the minimal number of forests required to cover the edges of the graph.}
A graph of arboricity at most $\alpha$
has the property that the average degree in any subgraph is at most $\alpha$~\cite{nash1961edge, nash1964decomposition}.
Constant arboricity families are a rich class, containing all minor-closed graph families,
bounded expansion graphs, and preferential attachment graphs. A classic result of Chiba and Nishizeki gives
an $O(n+m\alpha^{k-2})$ algorithm for exact counting of $k$-cliques in graphs of arboricity at most $\alpha$~\cite{ChNi85}.
Our primary motivation is to get a sublinear-time algorithm for approximating the number of $k$-cliques  on such graphs. We assume the standard query model for general graphs (refer to Chapter 10 of Goldreich's book~\cite{G17-book}), so that the algorithm can perform degree, neighbor and pair queries.\footnote{Let us exactly specify each query.
	\begin{inparaenum}[(1)]
		\item Degree queries: given $v \in V$, get the degree $d(v)$.
		\item  Neighbor queries: given $v \in V$ and $i \leq d(v)$
		get the $i\th$ neighbor of $v$.
		\item Pair queries: given vertices $u,v$, determine if $(u,v)$ is an edge.
\end{inparaenum}}

\subsection{Results}
Our main result is an algorithm for approximating the number of $k$-cliques,
whose complexity
depends on the arboricity. The algorithm is sublinear for $n_k = \omega(\alpha^{k-2})$,
(and we subsequently show that for smaller $n_k$, no sublinear algorithm is possible).

\begin{thm} \label{thm:main}
There exists an algorithm
that, given $n,\;k$, an approximation parameter $0<\eps<1$, query access to a graph $G$,
and an upper bound $\alpha$ on the arboricity of $G$,
outputs an estimate $\hnk$, such that with high constant probability (over the randomness of the algorithm), $$(1-\eps)\cdot \clk\leq \hnk\leq (1+\eps)\cdot \clk.$$
The expected
running time of the algorithm is \[\min\left\{\frac{n\alpha^{k-1}}{n_k},\; \frac{n}{n_k^{1/k}}+\frac{m \alpha^{k-2}}{n_k} \right\}\cdot\poly(\log n, 1/\eps, k^k),\]
and the expected  query complexity is the minimum between the expected running time and $O(m+n)$.
\end{thm}

Recall that $\alpha$ is always upper bounded by $\sqrt m$,  so that the bound in Theorem~\ref{thm:main} subsumes the result for general graphs~\cite{ERS18}.
Also observe that, for $n_k \gg \poly(\log n, 1/\eps, k^k)$-factor, this bound for approximate counting
improves that of Chiba and Nishizeki~\cite{ChNi85} for exact counting.

An application of Theorem~\ref{thm:main} for the
family $\mG$ of minor-closed graphs\footnote{A family of graphs is said to be minor-closed if it is closed under vertex removals, edge removals and edge contractions.} gives the following corollary.

\begin{cor} \label{cor:minor}
Let $\mG$ be a minor-closed family of graphs. There is an
algorithm that, given $n,\;k, \eps > 0$, and query access to $G \in \mG$, outputs a $(1\pm\eps)$-approximation of $n_k$ with high constant probability.
The expected running time of the algorithm is
\ifnum\conf=0
\[\frac{n}{n_k}\cdot  \poly(\log n, 1/\eps, k^k).\]
\else
~$(n/n_k) \cdot \poly(\log n, 1/\eps, k^k)$.
\fi
\end{cor}

Even for the special case of triangle counting in planar graphs, such
a result was not previously known. Ignoring the dependence on $\log n$ , $\eps$ and $k$, the bound of $O(n/n_k)$
is particularly pleasing.
We also prove that the bound of Theorem~\ref{thm:main} is nearly optimal for low arboricity graphs.


\vspace{-0.7ex}
\begin{thm} \label{thm:lb}
Consider the set $\mG$ of graphs of arboricity  at most $\alpha$.
Any multiplicative approximation algorithm that succeeds with constant probability on
all graphs in $\mG$ must make
\ifnum\conf=0
\[\Omega\left( \min\left\{\frac{n\alpha^{k-1}}{n_k},\; \frac{n}{n_k^{1/k}} \right\}+\min\left\{\frac{m (\alpha/k)^{k-2}}{n_k}, \; m \right\} \right)\]
\else
\\
~~~~~$\Omega\left( \min\left\{ \frac{n}{k \cdot n_k^{1/k}},\; \frac{n\alpha^{k-1}}{k^k \cdot n_k} \right\}+\min\left\{\frac{m (\alpha/k)^{k-2}}{n_k}, \; m \right\} \right)$
\fi
queries.
\end{thm}

\vspace{-1.7ex}
\subsection{Related Work} \label{sec:related}

Clique counting, and the special case of triangle counting have received
significant attention in a variety of models. We refer the interested reader
to related work sections of~\cite{ELRS} and~\cite{ERS18} for general references.
We will focus on algorithms for low arboricity graphs.

The starting point for such algorithms is the seminal work
of Chiba and Nishizeki, who give a $O(m\alpha^{k-2} + n)$ algorithm
for enumerating $k$-cliques in a graph of arboricity at most $\alpha$~\cite{ChNi85}.
The usual approach to exploit the arboricity is to use degree or degeneracy orientations,
and this method has appeared in a number of theoretical and practical
results on triangle and clique counting~\cite{Co09,SuVa11,BFN+14,FFF15,JS17,DBS18}.
Recent work by Kopelowitz at al. shows
that improving the $O(m\alpha)$  bound for triangle counting is 3-SUM hard~\cite{KPP16}.

Our work follows a line of work on estimating subgraph counts using sublinear algorithms.
The first results were average degree estimation results of
Feige~\cite{feige2006sums} and Goldreich and Ron~\cite{GR08}. These ideas were extended
by Gonen et al. to estimate star counts~\cite{GRS11}. This was the first paper
that looked at the problem of estimating triangles, albeit from a lower bound perspective.
Eden et al. gave the first sublinear algorithm for triangle estimation,
whose query complexity (ignoring $\eps$ and $\log n$ factors) was $O(n/n^{1/3}_3 + m^{3/2}/n_3)$~\cite{ELRS}. This result was
generalized by the authors for $k$-clique counting (as mentioned earlier)~\cite{ERS18}.

The relevance of arboricity for sublinear algorithms was discovered
in the context of estimating stars (or degree moments) in previous work
by the authors~\cite{ERS17}. In that work, standard lower bounds for estimating
degree moments could be avoided for low arboricity graphs, just as in Theorem~\ref{thm:main}.
Recent work of Eden at al. gives a sublinear (bicriteria) algorithm for property testing arboricity~\cite{ELR18}.

On the data mining side, Dasgupta et al. and Chierichetti et al.
consider sublinear algorithms for estimating the average degree, in weaker
models than the standard property testing model~\cite{DKS,ChDa+16}. These results require extra assumptions
on the graphs. Eden et al. build on the ideas developed in work mentioned earlier to get
a practical algorithm for estimating the degree distribution~\cite{EJP+18}.

There is a rich literature on sublinear algorithms for estimating other graph parameters
such as the minimum spanning tree, matchings, and vertex covers~\cite{DBLP:journals/siamcomp/ChazelleRT05, DBLP:journals/siamcomp/CzumajS09, DBLP:journals/siamcomp/CzumajEFMNRS05,nguyen2008constant, yoshida2009improved,DBLP:journals/tcs/ParnasR07,nguyen2008constant,DBLP:journals/talg/MarkoR09, yoshida2009improved, hassidim2009local, onak2012near}.

\vspace{-1.7ex}
\subsection{Organization of the paper}\label{subsec:org}
Our algorithm and its analysis are quite involved. 
In Section~\ref{sec:overview} we give a fairly elaborate (but informal)
 overview of our algorithm and the ideas behind it. After introducing some preliminaries and
 defining some central notions (in Sections~\ref{sec:prel} and~\ref{sec:assn3-weight}), we
 provide our algorithm and the main procedure it uses 
 \ifnum\conf=1
 (in Sections~\ref{sec:oracle-based} and~\ref{sec:is-active}). The full details of its analysis are given in Sections~\ref{sec:missing-proofs}--\ref{sec:final}. 
 \else
\tnote{ (in Sections~\ref{sec:oracle-based} and~\ref{sec:isactive-full}). We then finalize all the details in Section~\ref{sec:final} to achieve our main theorem.}
 \fi
 We end with the lower bound in Section~\ref{sec:lb}.

\section{Overview of the algorithm and the main ideas behind it} \label{sec:overview}

As we explain  below, our starting point is similar to the one applied in~\cite{ERS18} for approximately counting
the number of $k$-cliques in general graphs (and that of~\cite{ELRS}, for $k=3$).
However, in order to exploit the fact that the graphs we consider have bounded arboricity, we depart quite early
from the~\cite{ERS18} algorithm, and introduce a variety of new ideas.
For the sake of simplicity of the presentation, assume that $\alpha<n_k^{1/k}$ and  $\eps$ is a constant,
so that we aim for an upper bound of roughly $O(n\alpha^{k-1}/n_k)$ (recall that $m \leq \alpha n$). In what follows we refer to~\cite{ERS18} as ERS.

\subsection{Common starting point with ERS}\label{subsec:ERS}
Assume we uniquely and arbitrarily assign each $k$-clique to one of its vertices.
For a vertex $v$  let $\wt(v)$ denote the number of $k$-cliques assigned to it, where we refer to this value as the {\em weight\/} of $v$.
Consider sampling a set $\mR$ of vertices uniformly at random,\footnote{The algorithm may actually obtain a multiset,
but in this exposition, we abuse terminology and call it a `set'.}
and let $\wt(\mR)= \sum_{v \in \mR}\wt(v)$. Clearly, $\EX[\wt(\mR)]=\frac{n_k}{n}\cdot |\mR|$.
However, $\wt(\mR)$  might have a large variance.
For example, consider the case of $k=3$ and the wheel graph, where it is possible that the central vertex is assigned all the triangles.
Hence, we need an assignment rule that assigns almost all $k$-cliques, but minimizes the
number of $k$-cliques assigned to any vertex. Furthermore, the rule should be efficiently
computable. That is, given a vertex $v$ and a $k$-clique $C$, it should be easy to verify whether $C$ is assigned to $v$.
Assume for now that we have such an assignment rule, and that $\wt(\mR)$ is indeed close to its expected value.

The next step is to estimate $\wt(\mR)$. Let $E_{\mR}$ denote the set of edges incident to the vertices of $\mR$,
and assume that $|E_{\mR}|$ is close to its expected value $\frac{m}{n}\cdot |\mR|$.
In ERS, $\wt(\mR)$ is approximated by sampling uniform edges in $E_{\mR}$ and extending them to 
$k$-cliques. Consider first the (easy) case where all the vertices 
have degree  $O(\sqrt m)$. Then we can extend an edge $(u,v)$ for  $u \in \mR$ to a (potential) $k$-clique
by sampling  $k-2$ neighbors of $u$, each with probability roughly $1/\sqrt m$ (and checking whether we obtained a clique). The probability that this process yields a $k$-clique is roughly $\frac{\wt(\mR)}{|E_{\mR}| \cdot \sqrt m^{k-2}} \approx \frac{n_k}{m^{k/2}}$. By repeating the above process $O(m^{k/2}/n_k)$ times,\footnote{The observant reader may be worried that this requires knowing $m$ and $n_k$, where the former is not provided to the algorithm and the latter is just what we want to estimate. However, constant factor estimates of both suffice for our purposes. For $m$ this can be obtained using~\cite{ERS17}, and for $n_k$ this assumption can be removed by performing a geometric search.
For details see Section~\ref{sec:final}.}
we can get an estimate of $\wt(\mR)$ and thus of $n_k$ (assuming an efficient verification procedure for the assignment rule).
For the case when degrees are much larger than $\sqrt{m}$,
ERS gives a more complex procedure that extends edges to (potential) $k$-cliques.
In the end, each $k$-clique is still sampled with probability roughly $n_k/m^{k/2}$.

In our setting (where the arboricity is at most $\alpha$) the simple scenario discussed above of vertex degrees bounded by $O(\sqrt{m})$ corresponds to the case that all vertex degrees are $O(\alpha)$.
In such a case we can extend an edge to a (potential) $k$-clique in the same manner as ERS, and get that the success probability of sampling a $k$-clique is $\Omega\left(\frac{n_k}{m\alpha^{k-2}}\right)$. 
Unfortunately, it is not clear how to adapt the ERS approach for the unbounded-degrees case and obtain a dependence on $\alpha$ instead of $\sqrt{m}$.
Therefore, at this point, we depart from the approach of ERS.

\subsection{An iterative sampling process}\label{subsec:it-samp}
The ERS algorithm can be viewed as a three-step process. It first samples vertices, then samples edges (incident to the sampled vertices), and then (in one step) samples $k$-cliques 
 that are extensions of these edges. To get a complexity depending
 on the arboricity, we devise an iterative clique sampling process. In  iteration $t$,  we obtain a sample of $t$-cliques, based on the sample of $(t-1)$-cliques from the previous iteration.

It is crucial in our analysis to distinguish \emph{ordered} cliques from unordered cliques.
An unordered $t$-clique $T$ is a set of $t$ vertices $T=\{v_1, \ldots, v_t\}$ (such that every two vertices are connected),
while an ordered $t$-clique is a tuple of $t$ vertices $\vT=(v_1, \ldots, v_t)$ such that $\{v_1, \ldots, v_t\}$ is a clique.
We say that $\vT=(v_1, \ldots, v_t)$ {\em participates\/} in a clique $C$, if $\{v_1, \ldots, v_t\} \subseteq C$.
We also extend the (yet undefined) assignment rule to allow assigning $k$-cliques to ordered $t$-cliques for any $t\leq k$
(and not just to vertices, which is the special case of $t=1$).
For an ordered $t$-clique $\vT$, let $\wt(\vT)$ be the number of $k$-cliques that are assigned to $\vT$,
and for a set of ordered $t$-cliques $\mR$,  let $\wt(\mR)=\sum_{\vT \in \mR}\wt(\vT)$.
We defer the discussion of the assignment rule. For now we focus on the algorithm.

The algorithm starts by sampling a set of \dnote{$s_1$} ordered $1$-cliques (vertices), denoted $\mR_1$. Assume  that 
$\wt(\mR_1)\approx\frac{n_k}{n}\cdot s_1$. 
The algorithm next
 samples a set of $s_2$ ordered $2$-cliques (ordered edges), denoted $\mR_2$, incident to the vertices of $\mR_1$.
\dnote{For $t>2$,} the $t^{\rm th}$ iteration extends $\mR_{t}$ to $\mR_{t+1}$, as described next.


For an ordered $t$-clique $\vT$, 
let $d(\vT)$ be the degree of the {\em minimum-degree} vertex in $\vT$, and
for a set of ordered cliques $\mR$,
let $d(\mR) =\sum_{\vT \in\mR}d(\vT)$.
The sampling of the set $\mR_{t+1}$ is done by repeating the following $s_{t+1}$ times:
sample a clique $\vT$ in $\mR_t$ with probability proportional to $d(\vT)/d(\mR_t)$ and then
select a uniform neighbor of the least degree vertex in $\vT$. 
Hence, each $(t+1)$-tuple  that is an extension of an ordered $t$-clique in $\mR_t$ is sampled with probability $\frac{d(\vT)}{d(\mR_t)}\cdot \frac{1}{d(\vT)}=\frac{1}{d(\mR_t)}$. For each sampled $(t+1)$-tuple, the algorithm checks whether it is a $(t+1)$-clique, and if
so, adds it to $\mR_{t+1}$.  Suppose that the weight function (defined by the assignment rule) has the following property.
The weight $\wt(\mR_t)$ is 
the sum of the weights taken over all ordered $(t+1)$-cliques that are extensions of the ordered $t$-cliques in  $\mR_t$.
We can conclude that the expected value of $\wt(\mR_{t+1})$ is $\frac{\wt(\mR_t)}{d(\mR_t)}\cdot s_{t+1}$.

We need to get good upper bounds for $s_{t+1}$, ensuring that $\wt(\mR_{t+1})$ is concentrated around its mean.
Note that the probability of getting a $(t+1)$-clique is inversely proportional to $d(\mR_t)$. Thus, we need
good upper bounds on this quantity, to upper bound $s_{t+1}$. This is where the arboricity enters the picture.
We give a simple argument proving that $d(\mC_t)=\sum_{\vT\in \mC_t}d(\vT) =O(m\alpha^{t-1})$.
(Note that the case $t=3$ is precisely the Chiba and Nishizeki bound $\sum_{(u,v)\in E}\min\{d(u),d(v)\}=O(m\alpha)$~\cite{ChNi85}.)
We then show that $d(\mR_t)$ is bounded as a function of $d(\mC_t)$.

%
%
\subsection{Desired properties of the assignment rule}\label{subsec:assign}
Recall that we need to ensure that
with high probability, $\wt(\mR_1)$ is close to its expected value, which should be close to $\frac{n_k}{n}\cdot \dnote{s_1}$, and that
 for every $t\geq 1$, $\wt(\mR_{t+1})$ is  close to $\frac{\wt(\mR_t)}{d(\mR_t)} \cdot \dnote{s_{t+1}}$.
 In addition, we need to efficiently verify the assignment rule.
 We achieve this by defining an assignment rule that has the following properties.
 \smallskip

\begin{asparaenum}
\item\label{prop:wtV} $\wt(V)\approx n_k$. This ensures that the expected value of $\wt(\mR_1)$ is approximately
$\frac{n_k}{n}\cdot s_1$.~
\item\label{prop:extend-mRt}
 For every $t$, the sum of the weights taken over all ordered $(t+1)$-cliques that are extensions of the ordered $t$-cliques in
 $\mR_t$ equals $\wt(\mR_t)$. This ensures that for every $t$, 
 $\EX[\wt(\mR_{t+1})]=\frac{\wt(\mR_t)}{d(\mR_t)} \cdot \dnote{s_{t+1}}$.~
\item\label{prop:bounded} For every ordered $t$-clique $\vT$, $\wt(\vT)$ is not too large. This ensures that with high probability
$\wt(\mR_{t+1})$ is close to its expected value for all $t$, for a sufficiently large sample size \dnote{$s_{t+1}$ (which depends on
this upper bound on $\wt(\vT)$ as well as on $d(\mR_t)$)}.
~ 
\item\label{prop:verify} Given a $k$-clique $C$ and an ordered $t$-clique $\vT=(v_1, \ldots, v_t)$ such that  $\{v_1, \ldots, v_t\} \subseteq C$, we can efficiently determine if $C$ is assigned to $\vT$. This ensures that when we get the final set
    $\mR_k$ of ordered $k$-cliques, we can compute  its weight (and deduce an estimate of $n_k$).
\end{asparaenum}

\smallskip
We introduce key notions in the definition of such an assignment rule.

\subsection{Sociable cliques and the assignment rule}\label{subsec:sociable-over}
For an ordered $t$-clique $\vT$,  let $c_k(\vT)$ denote the number of $k$-cliques containing $\vT$.
 An ordered $t$-clique $\vT$ is called \emph{sociable} if $c_k(\vT)$ is above a threshold $\tau_t\approx \alpha^{t-1}$.
  Otherwise, the clique is called  \emph{non-sociable}.
For $k$-clique $C = \{v_1,\dots,v_k\}$, let $\mOC(C)$ be the set of all ordered $k$-cliques
corresponding to the $k!$ tuples inducing $C$.
Let $\mOC'(C)$ be the subset of $\mOC(C)$ that contains ordered
$k$-cliques in $\mOC(C)$ such that {\em all prefixes are non-sociable\/}.
Consider the assignment rule that assigns $C$ to the first (in lexicographic order) $\vC\in \mOC'(C)$ and to each of
its prefixes.

In a central lemma (see Lemma~\ref{lem:almost-good-P}) we prove that the number of $k$-cliques that are not assigned by this assignment rule to any ordered $k$-clique (and its prefixes) is relatively small. The proof relies on the sociability thresholds $\{\tau_t\}$ and the fact that the graph has  arboricity at most $\alpha$.
We note that ERS also defined the notion of sociable vertices (as vertices that participate in too many $k$-cliques). However, their argument for bounding the number of unassigned $k$-cliques was simpler, as they did define and account for sociable cliques for $t>1$.

The aforementioned assignment rule addresses Properties~\ref{prop:wtV} to~\ref{prop:bounded}.
We are left with Property~\ref{prop:verify} (and how it fits in the big picture)

\subsection{Verifying an assignment and costly cliques}\label{subsec:verify-costly}
Recall that in the last iteration of the algorithm, it has a set $\mR_k$ of ordered $k$-cliques,
and it needs to compute $\wt(\mR_k)$ (which can be translated to an estimate of $n_k$). Namely, for each ordered $k$-clique $\vC=(v_1, \ldots, v_k)$ in $\mR_k$, the algorithm needs to verify whether the corresponding $k$-clique $C=\{v_1, \ldots, v_k\}$ is assigned to $\vC$.  This requires to verify
whether $\vC$ and each of its prefixes is non-sociable. Furthermore, it requires verifying that $\vC$ is the first
such ordered $k$-clique (in $\mOC(C)$).

For an ordered $t$-clique $\vT$, consider the subgraph $G_{\vT}$ induced by the set of vertices that neighbor
every vertex in $\vT$. Observe that $c_k(\vT)$ equals the number of $(k-t)$-cliques in
$G_{\vT}$. Therefore, deciding whether $\vT$ is sociable amounts to deciding whether the number of $(k-t)$-cliques in the subgraph $G_{\vT}$ is greater than $\tau_t$.
Indeed this is like our original problem of estimating the number of $(k-t)$-cliques in a graph,
except that it is applied to a subgraph $G_{\vT}$ of our original graph $G$. Unfortunately, we do not have direct query access to such subgraphs. To illustrate this, consider the case of $t=1$ so that $\vT$ consists of single vertex $v$. While we can sample uniform vertices in the subgraph $G_{(v)}$, we cannot directly perform neighbor queries (without incurring a possibly large cost when simulating queries to $G_{(v)}$ by performing
queries to $G$).

However, we show that we can still follow the high-level structure of our iterative sampling algorithm
(though there are a few obstacles).
Specifically, we initialize $\mR_t = \{\vT\}$,
and for each $j = t,\dots,k-1$, we sample a set of ordered $(j+1)$-cliques $\mR_{j+1}$ given a set of ordered $j$-cliques $\mR_j$,
exactly as described in Section~\ref{subsec:it-samp}. The first difficulty that we encounter is the following.
The success probability of sampling an ordered $(j+1)$-clique that extends an ordered $j$-clique in $\mR_j$
is inversely proportional to $d(\mR_j)$. Unfortunately, here we cannot argue that with
high probability $d(\mR_j)$ 
can be upper bounded as a function of
$d(\mC_j)$ (which is $O(m\alpha^{j-1})$). 
The reason is that while the algorithm described in Section~\ref{subsec:it-samp} starts with
a uniform sample of vertices $\mR_1$ (that the following samples $\mR_j$ build on),
here we start with $\mR_t = \{\vT\}$ for an arbitrary $t$-clique $\vT$. 

 We overcome this obstacle by defining the notion of
\emph{costly} cliques. We say that an ordered $t$-clique $\vT$ is costly if for some $j\geq t$, $d(\mC_j(\vT))$ is too large,
where $\mC_j(\vT)$ is the set of $j$-cliques that $\vT$ participates in.
For such ordered $t$-cliques, we cannot efficiently verify whether they are sociable.
Thus, we modify our assignment rule so that costly cliques are not  assigned any $k$-clique (even if they are non-sociable).
We prove that the additional loss in unassigned $k$-cliques is small (see Claim~\ref{clm:rel-cost}) and that we can efficiently
determine if an ordered $t$-clique is costly. (The precise definition is slightly different - see Definition~\ref{def:costly}, and the last assertion is more subtle -- see Claims~\ref{clm:IsActive-nonsociable} and~\ref{clm:IsActive-sociable}.) 

So we start with $\mR_t = \{\vT\}$, apply the iterative process, and obtain a set of ordered $k$-cliques $\mR_k$
(that are all extensions of $\vT$). To determine if $\vT$ is sociable, we need to estimate $c_k(\vT)$,
i.e., the number of $(k-t)$ cliques in $G_{\vT}$.
Luckily, it suffices to make this decision approximately.
For the analysis to go through, it suffices to distinguish between the case that $c_k(\vT)$ is ``too large'', and the case that it is
``sufficiently small''.
Therefore, given $\mR_k$, the final decision (regarding the sociability of $\vT$) can be made just based on $|\mR_k|$.

\subsection{Summary of our main new ideas and where arboricity comes into play}\label{subsec:main-ideas}
The following are the main differences and new ideas as compared to ERS, with an emphasis on the role of bounded arboricity.
\begin{asparaenum}
\item We introduce an iterative sampling process that, starting from a uniform sample $\mR_1$ of vertices, creates intermediate samples $\mR_t$ of ordered $t$-cliques, until it obtains a sample of ordered $k$-cliques. Arboricity comes into play here since the probability of obtaining an ordered $(t+1)$-clique that can be added to $\mR_{t+1}$, is inversely proportional to $1/d(\mR_t)$, which
    in turn can be bounded  as a function of $\alpha$ (and $m$).
\item We introduce an assignment rule and corresponding weight function $\wt$ that ensures two properties.
(1) Almost every $k$-clique is assigned (to some ordered $k$-clique and all its prefixes), and (2)
no ordered clique is assigned too many $k$-cliques.
The former implies that $\wt(V)\approx n_k$.
    The latter implies that, in the iterative sampling process, each sample of larger ordered cliques ``maintains the weight'' (up to an appropriate normalization) of the previous sample.

    The arboricity $\alpha$ determines the sociability thresholds $\{\tau_t\}$ (above which an ordered clique is not assigned any
    $k$-clique). These thresholds are carefully chosen to ensure that in graphs with arboricity at most $\alpha$,
    the number of unassigned $k$-cliques is sufficiently small. These parameters directly affect the time complexity of the algorithm.

\item We show how the assignment rule can be verified. This translates to determining whether certain ordered cliques are sociable.
A key notion is that of costly cliques, whose sociability cannot be determined efficiently. Arboricity also plays a role in their definition and in the proof that the additional loss incurred by not assigning $k$-cliques to costly ordered cliques is small.
\end{asparaenum}

\section{Preliminaries}\label{sec:prel}

For integer $j$, the set $\{1,\dots,j\}$ is denoted by $[j]$.
For a pair of integers $i \leq j$, the set of integers $\{i,\dots,j\}$ is denoted by $[i,j]$.
For a multiset $S$, we use $|S|$ to denote the sum of multiplicities of the items in $S$.
Our algorithm gets parameters $k$ and $\eps$, where we assume that $\eps < 1/2k^2$ (or else we set $\eps = 1/2k^2$).

\label{par:notation}Let $G= (V,E)$ be a graph with $n$ vertices, $m$ edges, and arboricity 
$\alpha(G)$.
As noted in Section~\ref{sec:overview}, we distinguish between a $t$-clique, which is a set of $t$ vertices $T=\{v_1,\dots,v_t\}$
(with an edge between every pair of vertices in the set),
and an {\em ordered\/} $t$-clique, which is a $t$-tuple of $t$ distinct vertices $\vT=(v_1,\dots,v_t)$ such that
 $\{v_1,\dots,v_t\}$ is a clique. For an ordered $t$-clique $\vT = (v_1,\dots,v_t)$, we use
 $U(\vT)$ to denote the corresponding unordered $t$-clique $\{v_1,\dots,v_t\}$.
For  cliques (ordered cliques) of size $1$, that is, vertices, we may use $v$ instead of $\{v\}$ (respectively, $(v)$), and similarly for cliques of size $2$ (edges).
We let $\mC_t(G)$ denote the set of $t$-cliques in $G$, and
 $n_t(G) = |\mC_t(G)|$. For the set of ordered $t$-cliques in $G$ we use $\mOC_t(G)$.
 When $G$ is clear from the context, we use the shorthand $\mC_t$, $n_t$ and $\mOC_t$, respectively.

\begin{definition}[Clique's least degree vertex and neighbors]\label{def:least_deg_ver}
	For a clique (or ordered clique) $C$ we let $\Gamma(C)$ denote the set of neighbors of $C$'s
minimal-degree vertex (breaking ties by ids)
and let $d(C) = |\Gamma(C)|$.
	We refer to $d(C)$ as the {\sf degree of the (ordered) clique} and to $\Gamma(C)$ as its {\sf set of neighbors}.
For a set (or multiset) of cliques (or ordered cliques) $\mR$, we use the notation $d(\mR)$ for $\sum_{C\in \mR} d(C)$.
\end{definition}
We stress that $\Gamma(C)$ (and respectively, $d(C)$) does not refer to the union of neighbors of vertices in $C$, but only to the neighbor of a single designated vertex in $C$.

\smallskip
Throughout the paper we introduce various additional notations. For the aid of the reader, they all appear in
Table~\ref{tab:notation} in Appendix~\ref{app:prel}.

\ifnum\conf=1
The proofs of the next two claims are provided in Appendix~\ref{appendix:prel-proofs}.
\fi

\begin{claim} 
\label{clm:sum_min_d_cliques}
	For every $t$,
	\confeqn{d(\mC_t(G))  
              \leq 2m\cdot \alpha(G)^{t-1}.}
\end{claim}	
\ifnum\conf=0
\SumMindCliques
\fi


\begin{claim} \label{clm:d^k_leq_d^k-1}
For every $t \geq 2$,
\confeqn{n_t(G) \leq \frac{2\alpha(G)}{t} \cdot n_{t-1}(G)\;.}
\end{claim}
\ifnum\conf=0
\BoundNkProof
\fi

As a corollary of Claim~\ref{clm:d^k_leq_d^k-1}, we obtain.

\begin{corollary} \label{cor:nk_vs_nt}
	For every $1\leq t <k$,
	\confeqn{n_k(G) \leq \frac{t!}{k!}\cdot n_{t}(G)\cdot (2\alpha(G))^{k-t} \;.}
\end{corollary}

\newcommand{\twt}{\widetilde{\wt}}
\newcommand{\uwt}{\twt}
\newcommand{\thresdRi}{?}
\newcommand{\ABORT}{\textsf{ABORT}}
\newcommand{\wgt}{\textbf{Is-Assigned}}
\newcommand{\hwt}{\widehat{WT}}
\newcommand{\sone}{\frac{n \tau_1}{\twt_0}\cdot \frac{3\ln(2/\beta)}{\gamma^2}}
\newcommand{\setst}{\frac{d(\mR_t)\cdot \tau_{t+1}}{\twt_{t}}\cdot \frac{3\ln(2/\beta)}{\gamma^2}}

\newcommand{\op}{\mathcal{Q}^{\mA}}

\section{Weight functions and assignments} \label{sec:assn3-weight}

As explained in the overview of our algorithm, 
a central component in our approach
is a weight function defined over ordered cliques.
We shall be interested in a weight function that is {\em legal\/} in the following sense.

\begin{definition}[A legal weight function]\label{def:legal-weight}
	A weight function $\wt:\bigcup_{t=1}^k \mOC_t \to \mathbb{N}$ is  {\sf legal\/}  if it
	satisfies the following.
	\begin{compactenum}
		\item For every ordered $k$-clique $\vC$, $\wt(\vC) \in \{0,1\}$, and
		for every unordered $k$-clique $C$, there is at most one ordered $k$-clique $\vC$ such that $C = U(\vC)$
		and $\wt(\vC)=1$.
		\item For every $t \in [k-1]$ and for every ordered $t$-clique $\vT$, $\wt(\vT) = \sum_{\vT' \in \mOC_{t+1}(\vT)} \wt(\vT')$.
	\end{compactenum}
	For a multiset of ordered cliques $\mR$, we let $\wt(\mR) = \sum_{\vT\in \mR} \wt(\vT)$.
\end{definition}
By the above definition,
\begin{fact}\label{fact:wtP-ub}
	Let $\wt$ be a legal weight function. Then
	$\wt(V) \leq n_k$.
\end{fact}

\label{par:active} We next show how to define a weight function based on
a subset $\mA = \bigcup_{t=1}^k \mA_t$ such that $\mA_t \subseteq \mOC_t$, which we refer to as a subset of
{\em active\/} ordered cliques.
Referring to the notions introduced informally in Section~\ref{sec:overview}, the intention is that
active ordered cliques will be non-sociable and non-costly
 cliques (where these notions are defined formally in Definitions~\ref{def:sociable} and~\ref{def:costly}, respectively).
The weight function is closely linked to the notion of assigning $k$-cliques to ordered cliques
(as becomes clear in Definition~\ref{def:wgt}).
For now our goal is to define such a weight function that is legal, and such that we can
easily verify (based on $\mA$) whether an ordered $k$-clique has weight $1$ or $0$.
We would like to devise a weight function $\wt$ such that $\wt(V)$ is not much smaller than $n_k$, and that the weight of very ordered clique is appropriately bounded.
We later provide sufficient conditions on $\mA$, which ensure that these properties hold.

\label{par:Tleqj}In what follows, for a set $\mR$ of ordered cliques (in particular of the same size), we say that  $\vT \in \mR$ is {\em first\/} in $\mR$ if it is lexicographically first. Also, for an ordered $t$-clique $\vT$ and $j \leq t$, we use $\vT_{\leq j}$ to denote the ordered $j$-clique formed by the first $j$ elements in $\vT$.

Since we shall be interested in active ordered cliques such that all of their prefixed are also active, it will be useful to define the notion of fully active cliques.

\begin{definition}[Fully-active cliques] \label{def:fully-active}
Let  $\mA $ be a subset of  ordered cliques. 
An ordered $t$-clique $\vT$ is {\sf fully active} with respect to $\mA$,
if all of its prefixes belong to $\mA$. That is, $\vT_{\leq j} \in \mA$ for every $j \in [t]$.
We denote the subset of $t$-cliques that are fully active with respect to $\mA$ by $\mFA_t$.
\end{definition}

We are now ready to define our assignment rule. Recall that for an ordered $k$-clique $\vC=(v_1, \ldots, v_k)$, we use $U(\vC)=\{v_1, \ldots, v_k\}$ to denote its corresponding unordered clique. Hence, for an unordered $k$-clique $C$, $U^{-1}(C)$ is the set of $k!$ ordered $k$-cliques $\vC'$ such that $U(\vC')=C$.

\begin{definition}[Assignment and weight] \label{def:wgt}
	Let $\mA$ 
be a subset of ordered cliques.
For each $k$-clique $C$, if $U^{-1}(C) \cap \mFA_k \neq \emptyset$, then $C$ is
{\sf assigned}  (with respect to $\mA$) to the first ordered $k$-clique $\vC \in U^{-1}(C) \cap \mFA_k$,
and to each ordered $t$-clique $\vC_{\leq t}$ for $t \in [k-1]$.
Otherwise ($U^{-1}(C) \cap \mFA_k =\emptyset$), $C$ is {\sf unassigned}. That is, we assign the $k$-clique $C$ to the first ordered $k$-clique $\vC$ in $U(C)$ that is fully active and to all of its prefixes, if such an ordered clique $\vC$ exists. Otherwise we do not assign $C$ to any ordered clique.

\label{par:wt}	For each ordered $t$-clique $\vT$, we let $\wtP(\vT)$  denote the number of $k$-cliques that are assigned to $\vT$, and we refer to $\wtP(\vT)$ as the {\sf weight of $\vT$} (with respect to $\mA$).
	\end{definition}

Observe that by Definition~\ref{def:wgt}, an ordered $t$-clique $\vT$ is assigned some $k$-clique $C$ only if it is the $t$-prefix of some fully active (with respect to $\mA$) ordered $k$-clique $\vC$. 
This implies that $\vT$ is active and
hence, only  ordered cliques in $\mA$ can have non-zero weight.

The next claim follows from Definition~\ref{def:wgt}.
\begin{claim}
	For any subset $\mA$ of ordered cliques, $\wtP(\cdot)$ is a legal weight function.
\end{claim}

The following definition encapsulates what we require from the resulting weight function $\wt^\mA$.
\begin{definition}[Good  active subset] \label{def:goodCPC}
	For an approximation parameter $\eps$ and a vector of weight thresholds 
$\vec{\tau} = (\tau_1,\dots,\tau_{k})$,
	we say that a 
	subset $\mA$ of  ordered cliques
	is {\sf \pgood }
	if the following two conditions hold:
	\begin{asparaenum}
		\item For every $t \in [k]$ and for every ordered $t$-clique $\vT$, $\wtP(\vT) \leq \tau_t$. 
		\item $\wtP(V) \geq (1-\eps/2)n_k$.  
	\end{asparaenum}
	If only the first condition holds, then we say that 
	$\mA$ is {\sf $\vec{\tau}$-bounded}.
\end{definition}
As we shall discuss in more detail subsequently, we obtain the first item in Definition~\ref{def:goodCPC}
by ensuring that $\mA$ includes only ordered cliques
that do not participate in too many $k$-cliques.

\section{An oracle based algorithm}\label{sec:oracle-based}

In order to make the presentation more modular, 
we first present an {\em oracle-based\/} algorithm.
That is, we assume  the algorithm, \Approx, is given access to an oracle $\op$ for a
subset of active ordered cliques $\mA$: for any given ordered clique $\vT$, the oracle $\op$ returns whether
$\vT \in \mA$. The algorithm also receives an approximation parameter $\eps$, a confidence parameter $\delta$,
a ``guess estimate'' $\tnk$ of $n_k$, an estimate $\tm$ of $m$, and a vector of weight-thresholds $\vec{\tau}$.
Our main claim will roughly be
that if $\mA$ is \pgood\ (as defined in Definition~\ref{def:goodCPC}), $\tm \geq m/2$ and
$\tnk\leq n_k$, then with probability at least $1-\delta$ the algorithm outputs a $(1\pm \eps)$ approximation of $n_k$,
by approximating $\wtP(V)$. \dnote{For a precise statement, see Theorem~\ref{thm:approx-weight}.}
A constant-factor estimate $\tm$ of $m$ can be obtained by calling the moments-estimation algorithm  of~\cite{ERS17}, which is designed to work for bounded-arboricity graphs (applying it simply to the first moment).
We can  alleviate the need for the parameter  $\tnk$, by  relying on the search algorithm of~\cite{ERS18} \dnote{--
see Section~\ref{sec:final} for details.}


The algorithm \Approx\ starts by selecting a uniform sample of vertices ($1$-cliques). It then continues iteratively, where at the start of each iteration it has a sample of ordered $t$-cliques $\mR_t$. It sends this sample to the procedure \SampleSet, which returns a sample of ordered $(t+1)$-cliques $\mR_{t+1}$. The ordered cliques in $\mR_{t+1}$ are extensions of ordered cliques in $\mR_{t}$.  Once the algorithm reaches $t=k$, so that it has a sample of ordered $k$-cliques, it calls the procedure
\wgt\ on each ordered $k$-clique $\vC$ in $\mR_k$ to 
check whether it is assigned the unordered clique $U(\vC)$ (i.e., $\wtP(\vC)=1$).
Finally it returns an appropriately normalized version of the total weight of $\mR_k$.

\newcommand{\stthres}{ \frac{4\tm \alpha^{t-1}\cdot \tau_{t+1}}{\tnk}\cdot \frac{(k!)^2\cdot 3\ln(2/\beta)}{ \beta^{\tnote{t}}\cdot \gamma^2}}

\newcommand{\thresRk}{\frac{(k!)^2}{\beta^k} \cdot \frac{n_k\cdot \tau_k}{\tnk}\cdot  \frac{12\ln(2/\beta)}{\gamma^2}}

\def\ApproxCliques{
	\begin{figure}[htb!]
		\fbox{
			\begin{minipage}{0.95\textwidth}
				{\bf \Approx}($n,k,\alpha,\eps,\delta,\tnk,\tm,\vec{\tau},\op$) \label{alg:const-approx}
				\smallskip
				\begin{compactenum}
					\item Define $\mR_0 = V$ and $d(\mR_0) = n$.
					\item Set $\twt_0 = (1-\eps/2) \tnk$, $\beta=\delta/(3k)$ and $\gamma=\eps/(2k)$. \label{step:twt_0}
					\item Sample  $s_1=\left\lceil\sone\right\rceil$ vertices u.a.r. and let $\mR_1$ be the chosen multiset.
					\item For $t=1$ to $k-1$ do:
					\begin{compactenum}
						\item Compute $d(\mR_t)$ and set $\twt_t=(1-\gamma)\frac{\twt_{t-1} }{d(\mR_{t-1})} \cdot s_{t}$
						and $s_{t+1}=\left\lceil\setst \right\rceil$. \label{step:twt_t}
						\item If $s_{t+1} >\stthres$ then \textbf{abort}. \label{step:abort_st}
						\item Invoke \SampleSet$(t,\mR_t, s_{t+1})$
						and let $\mR_{t+1}$ be the returned mutliset. \label{step:invoke-sample-long}
					\end{compactenum}
					\item \label{step:abort_Rk} If $|\mR_k| > \thresRk$ then \textbf{abort}.
					\item \label{step:hnk}
					Let $\hnk= \frac{n \cdot d(\mR_1)\cldots d(\mR_{k-1})}{s_1 \cldots \cdot s_{k}} \cdot \sum_{\vC \in \mR_k} \wgt(\vC, \op)$.
					
					\item \textbf{Return}  $\hnk$.
				\end{compactenum}
				\vspace{3pt}
			\end{minipage}
		}
	\end{figure}
	
}

\ifnum\conf=1
For the sake of the exposition, here we provide a slightly simplified version of the algorithm.
In particular, we give approximate settings of the variables in the algorithm, and
use the notation $\eqsim$ 
for these approximate settings. We also removed two steps in which the algorithm aborts, which are
used in order to bound the complexity of the algorithm.
For full details (as well as a complete analysis) see Section~\ref{sec:missing-proofs}.

\begin{figure}[htb!]
	\fbox{
		\begin{minipage}{0.95\textwidth}
			{\bf \Approx}($n,k,\alpha,\eps,\delta,\tm,\tnk,\vec{\tau},\op$) \label{alg:const-approx-short}
			\smallskip
			\begin{compactenum}
				\item Define $\mR_0 = V$, $d(\mR_0) = n$ and set $\twt_0 \eqsim \tnk$.
%
				\item Sample  $s_1 \eqsim \frac{n \tau_1}{\tnk}$ vertices u.a.r. and let $\mR_1$ be the chosen multiset.
				\item For $t=1$ to $k-1$ do:
				\begin{compactenum}
					\item Compute $d(\mR_t)$ and set $\twt_t \eqsim \frac{\twt_{t-1}}{d(\mR_{t-1})}  \cdot s_{t}$
					and $s_{t+1} \eqsim \frac{d(\mR_t)\cdot \tau_{t+1}}{\twt_t}$. 
					\item \label{step:approx-sample-set} Invoke \SampleSet$(t,\mR_t, s_{t+1})$
					and let $\mR_{t+1}$ be the returned multiset. \label{step:approx-sample-set}
				\end{compactenum}
				\item \label{step:hnk-short}
				Let $\hnk= \frac{n \cdot d(\mR_1)\cldots d(\mR_{k-1})}{s_1 \cldots \cdot s_{k}} \cdot \sum_{\vC \in \mR_k} \wgt(\vC, \op)$.
				
				\item \textbf{Return}  $\hnk$.
			\end{compactenum}
			\vspace{3pt}
		\end{minipage}
	}
\end{figure}

\else

\ApproxCliques

\fi
\mnote{T: Move loop to get better success probability from finalizing to here}

The procedure \SampleSet\ (invoked in
\ifnum\conf=1
Step~\ref{step:approx-sample-set}
\else
Step~\ref{step:invoke-sample-long}
\fi
of \Approx), is presented next.
Given a multiset $\mR_t$ of ordered $t$-cliques, consider all $(t+1)$-tuples that
each corresponds to an ordered $t$-clique $\vT$ in $\mR_t$, and a neighbor $v$ of $\vT$.
(For the definition of the neighbors of
an ordered clique $\vT$ and its degree $d(\vT)$
refer to Definition~\ref{def:least_deg_ver}.)
\SampleSet\ samples uniformly from these tuples, and includes in $\mR_{t+1}$ those $(t+1)$-tuples that are
ordered $(t+1)$-cliques.
Constructing a data structures that supports such sampling  can be implemented in linear time in $d(\mR_t)$ (see e.g., \cite{walker1974new,walker1977efficient,marsaglia2004fast}).

 For a $t$-clique $T = \{v_1,\dots,v_t\}$, and $j \geq t$, we let $\mC_j(T)$  denote the set of $j$-cliques that $T$ participates in.
 That is, the set of  $j$-cliques $T'$ such that $T \subseteq T'$.
\label{par:mC_jT}For an ordered $t$-clique $\vT$ we use $\mC_{j}(\vT)$ as a shorthand for $\mC_{j}(U(\vT))$, and also say that $\vT$  participates in each  $\vT' \in \mC_j(\vT)$.
We let $\mOC_j(\vT)$  denote the set of ordered $j$-cliques that
are extensions of $\vT$.
For a multiset  of (ordered) $t$-cliques $\mR$, we use $\mC_j(\mR)$ to denote the union of $\mC_j(T)$ taken over
 all $\vT \in \mR$, where here in ``union'' we mean with multiplicity,\footnote{The formal term should be ``sum'', but since ``sum'' is usually used in the context of numbers, we prefer to use ``union''.}
and  $\mOC_j(\mR)$ is defined analogously.
We extend the definition of $\mC_j$ and $\mOC_j$ to $t$-tuples such that for a $t$-tuple $\vT$ that does not correspond to a $t$-clique, $\mC_j(\vT)$ and $\mOC_j(\vT)$ are mapped to the empty set. \label{par:tuple} Finally, for an ordered $t$-clique $\vT = (v_1,\dots,v_t)$ and a vertex $u$, we use
$(\vT,u)$ as a shorthand for the $(t+1)$-tuple $(v_1,\dots,v_t,u)$.

\newcommand{\setsi}{
\frac{d(\mR_i)\cdot \tau_i}{\underline{WT(\mR_i)}}\cdot \frac{3\ln(\delta/4)}{\eps^2}
}

\newcommand{\thresi}{\frac{}{}}
\begin{figure}[htb!]
	\fbox{
		\begin{minipage}{0.95\textwidth}
			{\bf \SampleSet($t,\mR_t,s_{t+1}$)} \label{alg:sample-set}
			\smallskip
			\begin{compactenum}
				\item Compute $d(\mR_t)$ and set up a data structure to sample each $\vT\in \mR_t$ with probability $d(\vT)/d(\mR_t)$. \label{step:ds}
				\item Initialize $\mR_{t+1}=\emptyset$.
				\item For $\ell = 1$ to $s_{t+1}$: \label{step:setsi}
				\begin{compactenum}
					\item Invoke the data structure to generate an ordered clique $\vT_\ell$.
                    \item Query degrees of vertices in $\vT_\ell$, and find a minimum degree vertex
                     $u \in \vT_\ell$. 
					\item Sample a random neighbor $v_\ell$ of $u$.
					\item If $(\vT_{\ell}, v_\ell)$ is an ordered $(t+1)$-clique,
                                 add it  to $\mR_{t+1}$.
				\end{compactenum}
			\item \textbf{Return} $\mR_{t+1}$.
			\end{compactenum}
			\vspace{3pt}
		\end{minipage}
	}
\end{figure}

In Claim~\ref{clm:sampleset} (see Section~\ref{sec:missing-proofs})  we prove that for an appropriate setting of the parameter $s_{t+1}$, the set $\mR_{t+1}$ returned by the procedure is ``typical'' with respect to the set $\mR_t$. Essentially we prove that with high probability, $\wt(\mR_{t+1})$ is close to its expected value, $\frac{\wt(\mR_t)}{d(\mR_t)}\cdot s_{t+1}$, and that $d(\mC_j(\mR_{t+1}))$ is not much larger than its expected value $\frac{d(\mC_j(\mR_{t}))}{d(\mR_t)}\cdot s_{t+1}$.

The procedure \wgt\ (invoked in
\ifnum\conf=1
Step~\ref{step:hnk-short}
\else
Step~\ref{step:hnk}
\fi
of \Approx)
decides whether a given ordered $k$-clique $\vC$ is assigned $U(\vC)$, i.e., whether
$\wtP(\vC) = 1$. This is done 
following Definition~\ref{def:wgt}, given access to an oracle for $\mA$.
In
\ifnum\conf=1
Section~\ref{sec:is-active}
\else
Section~\ref{sec:isactive-full}
\fi
we replace the oracle by an explicitly defined procedure.

\begin{figure}[htb!]
	\fbox{
		\begin{minipage}{0.95\textwidth}
			{\bf \wgt}($\vC, \op$) \label{proc:wgt}
			\smallskip
			\begin{compactenum}
                \item Let  $C = U(\vC)$.
                \item For each $\vC' \in U^{-1}(C)$, check whether $\vC'$ is fully active with respect to $\mA$, i.e., whether
                 $\vC\in \mFA_k$, by calling
                  $\op$ on every prefix $\vC'_{\leq t}$ for $t\in [k-1]$.
                \item If $\vC \in \mFA_k$ and it is the first ordered $k$-clique in $\mFA_k$, then   \textbf{return} 1, otherwise
                 \textbf{return} 0.
				\end{compactenum}
\vspace{3pt}
\end{minipage}
}
\end{figure}

\newcommand{\nkRi}{n_k\cdot \frac{s_1\cldots s_{i}}{d(\mR_0) \cldots d(\mR_{t-1})}}

\ifnum\conf=1
As stated at the start of this section, we show that if $\op$ is an oracle for a \pgood\ subset $\mA$, $\tm \geq m/2$  and $\tnk \leq n_k$,
then with probability at least $1-\delta$, \Approx\ returns an estimate
$\hnk$ such that $\hnk \in (1\pm\eps)n_k$. The precise statement regarding the correctness and complexity of \Approx\ as well
as its proof are provided in Section~\ref{sec:missing-proofs}.

\section{Implementing an oracle for good subset $\mA$}\label{sec:is-active}
\label{par:c_k(T)}
In this section we describe a (randomized) procedure named \IsActive, that implements an oracle for a subset $\mA$, where with
high probability, $\mA$ is \pgood\ for an appropriate setting of $\vec{\tau}$.
For an ordered $t$-clique $\vT$, let
$c_k(\vT)$ denote the number of $k$-cliques that $\vT$ participates in (that is, $c_k(\vT) = |\mC_k(\vT)|$).
The procedure aims at determining whether $c_k(\vT)$ is larger than $\tau_t$, in which case it is not included in $\mA$.
This ensures that the first item in Definition~\ref{def:goodCPC} holds, since  $\wtP(\vT) \leq c_k(\vT)\leq \tau_t$.
On the other hand, the setting of $\vec{\tau}$ is such that despite not including in $\mA$ all ordered cliques $\vT$ for which
$c_k(\vT)$ is above the allowed threshold, we can still prove that the second item in Definition~\ref{def:goodCPC} holds as well.
(See Definition~\ref{def:tau_i} for the exact setting of $\vec{\tau}$.)

In order to give the idea behind the procedure \IsActive, consider the special case that $t=1$ and $T=\{v\}$
for some vertex $v$. Observe that $c_k(v)$ is the number of $k$-cliques in the subgraph induced by $v$ and its neighbors.
Roughly speaking, for $i=1$, the procedure \IsActive\ works similarly to the algorithm \Approx, subject to setting $\mR_1 = \{v\}$.
Namely, starting from $t=1$, and
using the procedure \SampleSet, it iteratively selects a sample $\mR_{t+1}$ of ordered $(t+1)$-cliques, given a sample
$\mR_t$ of ordered $t$-cliques.
Once it obtains $\mR_k$ it uses $|\mR_k|$ to decide whether 
$c_k(v)$ is too large (and hence should not be included in $\mA$).
As opposed to \Approx, here we do not ask for a precise estimate of $c_k(v)$, and hence this decision is simpler. 
In addition, the invariant that we would like to maintain is that the average number of
$k$-cliques that 
ordered cliques in $\mR_{t+1}$ participate in, does not deviate by much from the average number for $\mR_t$.
The procedure generalizes to $i>1$ by setting $\mR_i = \big\{\vI\big\}$ and starting the sampling process with $t=i$.

Observe that the procedure \IsActive\ may exit (in Step~\ref{step:isactive-exit}) before reaching $t=k-1$
if $s_{t+1}$ is above a certain threshold.
This early exit (with an output of \textsf{Non-Active})
 addresses the case that $\vI$ is ``costly''
(for an informal discussion see Section~\ref{subsec:verify-costly}, and for a formal one see Section~\ref{subsec:costly}).

As in the case of the algorithm \Approx, here we give a slightly simplified version of \IsActive.
For full details of \IsActive\ and its analysis, see Section~\ref{sec:isactive-full}.

\begin{figure}[htb!]
	\fbox{
		\begin{minipage}{0.95\textwidth}
			{\bf \IsActive}($i,\vI,k,\alpha,\eps,\delta,\tnk,\tm, \vec{\tau}$) \label{alg:is-active-short}
			\smallskip
			\begin{compactenum}
					\item Let $\mR_i=\big\{\vI\big\}$ and $\twt_i \eqsim\tau_i$.
					\item For $t=i$ to $k-1$ do:
					\begin{compactenum}
						\item Compute $d(\mR_t)$.
						\item For $t>i$ set
						$\twt_t \eqsim \frac{\twt_{t-1} }{d(\mR_{t-1})} \cdot s_t$ and
						$s_{t+1} \eqsim \frac{d(\mR_t)\cdot \tau_{t+1}}{\twt_t}$.
						\item \label{step:isactive-exit} If $s_{t+1}$ is larger than
$\frac{\tm\alpha^{t-1}\cdot \tau_{t+1}}{\tnk}$, then \textbf{return} \textsf{Non-Active}.
						\item Invoke 						\SampleSet$(t,\mR_t, s_{t+1})$ and let $\mR_{t+1}$ be the returned multiset.
					\end{compactenum}
					\item Set $\hck(\vI):= \frac{d(\mR_i)\cdot \ldots\cdot d(\mR_{k-1})}{s_{i+1}\cdot\ldots\cdot  s_{k}}\cdot |\mR_k|$.
					\item If $\hck(\vI) \leq \tau_i/4$ then \textbf{return} {\sf Non-Active}. Otherwise, \textbf{return} {\sf Active}.
			
				\end{compactenum}
			\vspace{3pt}
		\end{minipage}
		}
\end{figure}

By replacing the calls to an oracle $\op$ in the algorithm \Approx\ with calls to \IsActive,
we obtain an algorithm that with high probability computes a $(1\pm \eps)$-estimate of $n_k$ (conditioned on
$\tm \geq m/2$ and
$\tnk < n_k$, where both assumptions can be removed).
For full details, see Section~\ref{sec:final}.

\fi

\ifnum\conf=1
\section{Filling the details for the oracle-based algorithm}
\else
\subsection{Analysis of \SampleSet\ }
\fi

\label{sec:missing-proofs}
In this and the next section
we shall make use of the following version of the multiplicative Chernoff bound. 
	Let $\chi_1,\dots,\chi_m$  be 
independent random variables taking values in $[0,B]$, such that for every $i$, $\EX[\chi_i]=p$. Then
\ifnum\conf=0
\[ \Pr\left[\frac{1}{m}\sum\limits_{i=1}^m \chi_i > (1+ \gamma)\mu \right] < \exp\left( -\frac{\gamma^2 \mu m}{3B}\right),
\text{  and  }
 \Pr\left[\frac{1}{m}\sum\limits_{i=1}^m \chi_i < (1- \gamma)\mu \right] < \exp\left( -\frac{\gamma^2 \mu m}{2B}\right)\;.\]
\else
\small
\begin{align*}
 \Pr\left[\frac{1}{m}\sum\limits_{i=1}^m \chi_i > (1+ \gamma)p \right] < \exp\left( -\frac{\gamma^2 p m}{3B}\right) \; \mbox{ and }\;
 \Pr\left[\frac{1}{m}\sum\limits_{i=1}^m \chi_i < (1- \gamma)p \right] < \exp\left( -\frac{\gamma^2 p m}{2B}\right)\;.
\end{align*}
\normalsize
\fi

\ifnum\conf=0
We shall also use the following notation. For an ordered $t$-clique $\vT$, let
$c_k(\vT)$ denote the number of $k$-cliques that $\vT$ participates in (that is, $c_k(\vT) = |\mC_k(\vT)|$).
\fi

We start by analyzing the procedure \SampleSet.

\begin{claim}[\SampleSet\ correctness]
	\label{clm:sampleset}
	Consider an invocation of \SampleSet\ with parameters
	$(t, \mR_t,s_{t+1})$.
	The following holds for the multiset $\mR_{t+1}$ that is returned by \SampleSet.
	\begin{compactenum}
		\item \label{item:sampleset-weight}
		Let $\wt$ be a legal weight function over ordered cliques such that $\wt(\vT') \leq \tau_{t+1}$ for every ordered
		$(t+1)$-clique $\vT'$, and let $s_{t+1} =\frac{d(\mR_t)\cdot \tau_{t+1}}{\twt_t}\cdot \frac{3\ln(2/\beta)}{\gamma^2}$.
		\begin{compactenum}
			\item \label{item:sampleset-weight1} If $\twt_t \leq \wt(\mR_t)$,
			then  $\wt(\mR_{t+1}) \in (1\pm \gamma) \cdot \frac{\wt(\mR_{t})}{d(\mR_{t})}\cdot s_{t+1}$
			with probability at least $1-\beta$.
			\item \label{item:sampleset-weight2} If $\twt_t > \wt(\mR_t)$,
			then  $\wt(\mR_{t+1}) \leq (1+ \gamma) \cdot \frac{\twt_t}{d(\mR_{t})}\cdot s_{t+1}$
			with  probability at least $1-\beta$.
		\end{compactenum}
		\item \label{item:sampleset-ck}
		$c_k(\mR_{t+1})\leq \frac{(k-t)}{\beta} \cdot \frac{c_k(\mR_{t})}{d(\mR_{t})}\cdot s_{t+1}$
		with probability at least $1-\beta$  (for any setting of $s_{t+1}$).
		\item\label{item:sampleset-degrees}  For every $t+1 \leq j <k$,
		$d(\mC_j(\mR_{t+1})) \leq  \frac{(k-t-1)^2}{\beta} \cdot \frac{d(\mC_j(\mR_{t}))}{d(\mR_{t})}\cdot s_{t+1}$.
		with probability at least $1-\beta$ (for any setting of $s_{t+1}$).
	\end{compactenum}
	The query complexity and running time of the procedure are $O\left(|\mR_t| + t\cdot s_{t+1}\right)$. 
\end{claim}

\begin{proof}
We start by proving
Item~\ref{item:sampleset-weight} in the claim
(regarding the weight of $\mR_{t+1}$).
For the sake of the proof, we extend the definition of the weight function $\wt(\cdot)$ to $t$-tuples,
so that the weight of any tuple that does not correspond to an ordered clique is 0.
For any legal $\wt$,
\begin{equation}
\EX_{\vT \in \mR_t, v \in \Gamma(\vT)}[\wt((\vT,v))]=\sum_{\vT \in\mR_t} \sum_{v \in \Gamma(\vT)}\frac{d(T)}{d(\mR_t)}\cdot \frac{1}{d(\vT)} \cdot \wt((\vT,v)) = \frac{\wt(\mR_t)}{d(R_t)}\;.
\end{equation}

By the premise of this item,
$\wt(\vT')\leq \tau_{t+1}$ for every ordered $(t+1)$-clique $\vT'$.
For the first part (Item~\ref{item:sampleset-weight1}: $\twt_t \leq \wt(\mR_t)$) we have that
$s_{t+1} \geq \frac{d(\mR_t)\cdot \tau_{t+1}}{\wt(\mR_t)}\cdot \frac{3\ln(1/\beta )}{\gamma^2}$, so by the multiplicative Chernoff bound  we get
\begin{equation}
\Pr\left[ \left\lvert \frac{1}{s_{t+1}}\sum_{\ell=1}^{s_{t+1}} \wt((\vT_\ell,v_\ell)) - \frac{\wt(\mR_t)}{d(\mR_t)} \right\rvert > \gamma\cdot  \frac{\wt(\mR_t)}{d(\mR_t)} \right] < 2\exp\left(-\frac{\gamma^2 \cdot \frac{\wt(\mR_t)}{d(\mR_t)}  \cdot s_{t+1}}{3\tau_{t+1}}\right) < \beta\;.
\end{equation}
Therefore, in this case, with probability at least $1-\beta$,
$\wt(\mR_{t+1})= \sum_{\ell=1}^{s_{t+1}} \wt((\vT_\ell,v_\ell))\in (1\pm \gamma) \cdot \frac{\wt(\mR_t)}{d(\mR_t)}\cdot s_{t+1}$, as required.

For the second part  (Item~\ref{item:sampleset-weight1}: $\wt(\mR_t) < \twt_t$), we can upper bound the probability that
$\wt(\mR_{t+1}) > (1+\gamma)\cdot \frac{\twt_t}{d(\mR_t)}\cdot s_{t+1}$ by the probability that this event occurs for
a multiset $\mR'$ (of ordered $t$-cliques) that satisfies $\wt(\mR') = \twt_t$, and this item too follows by
applying the multiplicative Chernoff bound (given the setting of $s_{t+1}$).

We next prove Item~\ref{item:sampleset-ck}.
Observe that since $c_k(\mOC_{t+1}(\mR_t)) = (k-t)\cdot c_k(\mR_t)$,
	\begin{equation}
\EX_{\vT \in \mR_t, v \in \Gamma(\vT)}[c_k\big((\vT,v)\big)] =\frac{ c_k(\mOC_{t+1}(\mR_t))}{d(\mR_t)}
= \frac{(k-t) \cdot c_k(\mR_t)}{d(\mR_t)}\;.
\end{equation}
	By Markov's inequality,
	\begin{equation}
\Pr\left[ \frac{1}{s_{t+1}}\sum_{\ell=1}^{s_{t+1}} c_k\big((\vT_\ell,v_\ell)\big)> \frac{\tnote{1}}{\beta}\cdot \frac{(k-t)\cdot c_k(\mR_t)}{d(\mR_t)} \right] < \frac{\beta}{k-t}.  \label{eqn:ck_ub}
	\end{equation}
We get that
	with probability at least $1-\beta$,
	\begin{equation}
c_k(\mR_{t+1}) \leq  \frac{(k-t)}{\beta}\cdot \frac{c_k(\mR_t)}{d(\mR_t)}\cdot s_{t+1}\;,
	\end{equation}		
as claimed.	

We now turn to prove Item~\ref{item:sampleset-degrees}.
For every $j\in\{t+1,\dots,k-1\}$,
\begin{equation}
\EX_{\vT \in \mR_t, v \in \Gamma(\vT)}[d\big(\mC_j((\vT,v))\big)] =\frac{d\big(\mC_j( \mOC_{t+1}(\mR_t))\big)}{d(\mR_t)}=\frac{(j-t)\cdot d(\mC_j( \mR_t))}{d(\mR_t)}\;,
\end{equation} 
and by Markov's inequality,
\begin{equation}
\Pr\left[ \frac{1}{s_{t+1}}\sum_{\ell=1}^{s_{t+1}} d(\mC_j((T_\ell,v_\ell)))> \frac{k-t-1}{\beta}\cdot \frac{(j-t)\cdot d(\mC_j(\mR_t))}{d(\mR_t)} \right] < \frac{\beta}{k-t-1}.  \label{eqn:dkCj_ub}
\end{equation}
Hence, by taking a union bound over all $j$'s in $[t+1, k-1]$,
it holds that with probability at least $1-\beta$,
\begin{equation}
d(\mC_j(\mR_{t+1})) \leq \frac{(k-t-1)^2}{\beta}\cdot \frac{d(\mC_j(\mR_t))\cdot s_{t+1}}{d(\mR_t)}
\;.
\end{equation}
The data structure described in Step~\ref{step:ds} can be implemented in linear time in $d(\mR_t)$ (see e.g., \cite{walker1974new,walker1977efficient,marsaglia2004fast}). Hence, the claim regarding the running time of \SampleSet\ follows directly from the description of the procedure (where the factor $t$ arises from performing $t$ pair queries between the sampled vertex $v$ and each vertex of $\vT$).
\end{proof}

\ifnum\conf=1
In order to prove the main theorem regarding the correctness and complexity of \Approx, we first give the full version of the algorithm and establish several claims. Observe that the algorithm may abort without returning an output (in either Step~\ref{step:abort_st} or
in Step~\ref{step:abort_Rk}). This allows us to upper bound the complexity of the algorithm, where we prove that the algorithm aborts with small probability (conditioned on $\tm$ being a constant-factor estimate of $m$).

\ApproxCliques

\else
\subsection{Analysis of \Approx}

\fi
\dnote{In this subsection we prove the following theorem regarding the performance of \Approx.}

\begin{thm}
	\label{thm:approx-weight}
	Consider an invocation of \Approx\ with parameters $(n,k,\alpha,\eps,\delta,\tm,\tnk,\vec{\tau},\op)$.
\begin{compactenum}
	\item If 
	$\mA$ is \pgood, $\tm \geq m/2$
	and
	$\tnk \leq n_k$, 
	then with probability at least $1-\delta$,
	\Approx\ returns an estimate
	$\hnk$ such that $\hnk \in (1\pm\eps)n_k$.  
	\item If 
	$\mA$ is
	$\vec{\tau}$-bounded, $\tm \geq m/2$ and
	$\tnk > n_k$, then with probability at least $1-\delta$, \Approx\ returns an estimate
	$\hnk$ such that
	$\hnk \leq (1+\eps)\tnk$.
\end{compactenum}
	
	Furthermore, let $\rho(\op)$ be an upper bound on the running time of $\op$ (per call) and let $q(\op)$ be an upper bound on its query complexity.
	The running time of \Approx\ is
	$$O\left(\frac{n \tau_1}{\tnk} + \frac{k^{3k}}{\delta^{2k}}\cdot \frac{\tm}{\tnk}\cdot
	\sum_{t=2}^{k-1} (\alpha^{t-2}\cdot \tau_t)
	+  \frac{k^{3k}}{\delta^{k}}\cdot   \frac{n_k}{\tnk}\cdot  \tau_k \cdot \rho(\op)
	\right)\cdot \frac{3^k\cdot \ln(k/\delta)}{\eps^2}\;, $$
	and the query complexity of the algorithm is upper bounded by the same  expression
 with $q(\op)$ in place of $\rho(\op)$.
	\end{thm}

\dnote{In order to prove Theorem~\ref{thm:approx-weight}, we first establish several claims.}

\begin{claim}\label{clm:all-R-typical}
	Consider an invocation of \Approx\ with parameters $(n,k,\alpha,\eps,\delta,\tnk,\vec{\tau},\op)$.
	If the algorithm did not abort, then the following hold.
\begin{compactenum}
\item\label{item:Approx-weight1}	
If $\tnk \leq n_k$ and $\mA$ is \pgood,
then with probability at least $1-\delta/3$,  for every $t \in [k]$, the sample of ordered $t$-cliques $\mR_t$ satisfies
$\wtP(\mR_{t}) \in (1\pm \gamma) \cdot \frac{\wtP(\mR_{t-1})}{d(\mR_{t-1})}\cdot s_{t}$.
\item \label{item:Approx-weight2} If $\tnk > n_k$ and $\mA$ is $\vec{\tau}$-bounded,
then with probability at least $1-\delta/3$,  for every $t \in [k]$, the sample of ordered $t$-cliques $\mR_t$
satisfies
$\wtP(\mR_{t}) \leq (1+ \gamma) \cdot \frac{\max\{\wtP(\mR_{t-1}),\twt_{t-1}\}}{d(\mR_{t-1})}\cdot s_{t}$.
\item\label{item:Approx-degrees} With probability at least $1-\delta/3$, for every $t \in [k]$, the sample of ordered $t$-cliques $\mR_t$
is such that for every $t \leq j <k$,
    $d(\mC_j(\mR_{t})) \leq  \frac{\tnote{(k-t-1)^2}}{\beta} \cdot \frac{d(\mC_j(\mR_{t-1}))}{d(\mR_{t-1})}\cdot s_{t}$.
\end{compactenum}
\end{claim}

\begin{proof}
In what follows, we say that $\mR_t$ is {\em weight-typical\/} with respect to $\wtP$ and $\mR_{t-1}$ if
$\wtP(\mR_{t}) \in (1\pm \gamma) \cdot \frac{\wtP(\mR_{t-1})}{d(\mR_{t-1})}\cdot s_{t}$.
We say that  $\mR_t$ is {\em degrees-typical\/} with respect to $\mR_{t-1}$ if
$d(\mC_j(\mR_{t})) \leq  \frac{(k-t)}{\beta} \cdot \frac{d(\mC_j(\mR_{t-1}))}{d(\mR_{t-1})}\cdot s_{t}$.

For 
Item~\ref{item:Approx-weight1},
we start by proving that (conditioned on the premise of the item) $\mR_1$ is weight-typical with respect to $\wtP$ and $\mR_0=V$ with probability at least
$1-\beta$ (recall that we defined $d(\mR_0)=n$, and that $\beta = \delta/(3k$)).
The argument is indeed very similar to the one used to prove 
Item~\ref{item:sampleset-weight1}
of Claim~\ref{clm:sampleset}, but the sampling process is slightly different, and hence we need to give a separate proof.
Clearly, $\EX_{v \in V}[\wtP(v)] = \frac{\wtP(V)}{n}$. Also, since 
$\mA$ is \pgood, it holds that for
every vertex $v \in V$, $\wtP(v)\leq \tau_1$ and that $\wtP(V) \in [(1-\eps/2)n_k, n_k]$.
Therefore, if $\tnk \leq n_k$, by the setting of $\twt_0=(1-\eps/2)\tnk$ in Step~\ref{step:twt_0},
we have that
$\twt_0 \leq \wtP(V)\label{eqn:lb_wt} \;.$
By the multiplicative Chernoff bound and the setting of $s_1 = \sone$,
\begin{equation}
\Pr\left[ \left\lvert \frac{\wtP(\mR_1)}{s_{1}} - \frac{\wtP(V)}{n} \right\rvert > \gamma\cdot  \frac{\wtP(V)}{n} \right] < 2\exp\left(-\frac{\gamma^2 \cdot \frac{\wtP(V)}{n}  \cdot s_{1}}{3\tau_{1}}\right) < \beta\;.
\end{equation}
Therefore, $\wtP(\mR_1) \in (1\pm \gamma)\cdot \frac{\wtP(V)\cdot s_1}{n}$ with probability at least $1-\beta$,
as claimed.

We next show that for $t >0$, conditioned on
$\mR_i$ being weight-typical with respect to $\wtP$ and $\mR_{i-1}$ for every $i=1,\dots,t$,
the sample $\mR_{t+1}$ is weight-typical with respect to $\wtP$ and $\mR_t$ with probability at
least $1-\beta$.
Observe first that by the above conditioning,
\begin{equation}
\wtP(\mR_{t}) \in
(1\pm \gamma)^t \cdot \frac{\wtP(V) \cdot s_1 \cldots s_t}{n \cdot d(\mR_{1}) \cldots d(\mR_{t-1})}\;.
\end{equation}
%
Together with Equation~\eqref{eqn:lb_wt} and the setting of $\twt_t$ in Step~\ref{step:twt_t}, this implies that $\twt_t \leq \wtP(\mR_t)$.
Therefore,
we can apply 
Item~\ref{item:sampleset-weight1}
of Claim~\ref{clm:sampleset} and get that
the procedure \SampleSet\ (when called with $t$, $\mR_t$ and $s_{t+1}$)
returns a multiset $\mR_{t+1}$ that
with probability at least $1-\beta$
is weight-typical with respect to $\wtP$ and $\mR_{t}$.
Item~\ref{item:Approx-weight1}
of this claim follows by taking a union bound over all $t$  
(and recalling that $\beta = \delta/(3k)$).


The proof of Item~\ref{item:Approx-weight2} is similar to the proof of Item~\ref{item:Approx-weight1},
except that here we need to (also) apply Item~\ref{item:sampleset-weight2} of Claim~\ref{clm:sampleset}.
Similarly to the proof of Item~\ref{item:Approx-weight1}, we first consider $\wtP(\mR_1)$.
Since $\tnk > n_k$, here we have that
$\wtP(\mR_1) \leq (1+\gamma) \cdot \frac{\tnk\cdot s_1}{n}$ with probability at least $1-\beta$.
For each $t > 1$, depending on whether $\twt_t \leq \wtP(\mR_t)$ or $\twt_t > \wtP(\mR_t)$,
we can apply either Item~\ref{item:Approx-weight1} or Item~\ref{item:Approx-weight2}.
In the first case we get that $\wtP(\mR_{t+1}) \leq (1+\gamma)\cdot \frac{\wtP(\mR_t)}{d(\mR_{\tnote{t}})}\cdot s_{\tnote{t+1}}$
with probability at least $1-\beta$, and in the second case that
$\wtP(\mR_{t+1}) \leq (1+\gamma)\cdot \frac{\twt_t}{d(\mR_{\tnote{t}})}\cdot s_{\tnote{t+1}}$ with probability at least $1-\beta$,
and Item~\ref{item:Approx-weight2} follows (by taking a union bound over all $t$).

We now turn to 
Item~\ref{item:Approx-degrees}.
We first prove that with probability at least $1-\beta$, the sample $\mR_1$
is degrees-typical (with respect to $\mR_0 = V$).
 For every $j\in[2,k-1]$,
\begin{equation}
\EX_{v \in V}[d(\mC_j(v))] =\frac{d(\mC_j(V))}{n}\;,
\end{equation}
and by Markov's inequality,
\begin{equation}
\Pr\left[ \frac{d(\mC_j(\mR_1))}{s_{1}} > \frac{k-2}{\beta}\cdot \frac{d(\mC_j(V))}{n} \right] < \frac{\beta}{k-2}.  
\end{equation}
By taking a union bound over all $j$'s in $\{2,\dots,k-1\}$,
it holds that with probability at least $1-\beta$,
\begin{equation}
d(\mC_j(\mR_{1})) \leq \frac{k-2}{\beta}\cdot \frac{d(\mC_j(V))}{n}\cdot s_{1}
\end{equation}
for every such $j$.
It follows that with probability at least $1-\beta$, the multiset $\mR_1$ is degrees-typical with respect to 
$\mR_0=V$. For $t\in \{2,\dots,k\}$ we apply 
Item~\ref{item:sampleset-degrees} of
Claim~\ref{clm:sampleset} and 
Item~\ref{item:Approx-degrees} of
this claim follows
by taking a union bound over all $t$. 
%
\end{proof}


\begin{claim}\label{clm:nk_d(R)}
If $\tm \geq m/2$ and for every $t\in [k]$ the sample $\mR_t$
is such that $d(\mC_j(\mR_{t})) \leq  \frac{(k-t-1)^2}{\beta} \cdot \frac{d(\mC_j(\mR_{t-1}))}{d(\mR_{t-1})}\cdot s_{t}$
for every $t \leq j <k$, where $s_t$ is as defined in Step~\ref{step:twt_t} of Algorithm \Approx,
then for every $t\in [k]$,
\[s_{t+1} \leq \stthres\;.
\]
\end{claim}

\begin{proof} 
Consider any $t \in [k]$, and observe that $d(\mR_t) = d(\mC_t(\mR_t))$.
By the premise of the claim 
and the setting of $\twt_t$ in \Approx,
\begin{eqnarray}
\frac{d(\mR_t)}{\twt_t} &\leq &
    \frac{\frac{(k-t-1)^2}{\beta}\cdot \frac{d(\mC_t(\mR_{t-1}))}{d(R_{t-1})}\cdot s_t }
       {(1-\gamma)\cdot \frac{\twt_{t-1}}{d(\mR_{t-1})} \cdot s_t}
\;=\;\frac{(k-t-1)^2}{(1-\gamma)\cdot\beta}\cdot \frac{d(\mC_t(\mR_{t-1}))}{\twt_{t-1}} \nonumber \\
& \leq&
\frac{(k-t-1)^2}{(1-\gamma)\cdot\beta}\cdot  \frac{\frac{(k-t)^2}{\beta}\cdot \frac{d(\mC_t(\mR_{t-2}))}{d(\mR_{t-2})} \cdot s_{t-1}}{(1-\gamma)\cdot \frac{\twt_{t-2}}{d(\mR_{t-2})}\cdot s_{t-1}}
\;= \;\frac{(k-t-1)^2\cdot (k-t)^2}{(1-\gamma)^2\cdot\beta^2}\cdot \frac{d(\mC_t(\mR_{t-2}))}{\twt_{t-2}} \nonumber \\
&\leq & \ldots \nonumber\\
&\leq&
\frac{(k-1)^2\cldots (k-t-1)^2}{\tnote{(1-\gamma)^{t}\cdot\beta^{t}}}\cdot \frac{d(\mC_t(\mR_{0}))}{\twt_0}
\;\leq\;\frac{(k!)^2}{\beta^{\tnote{t}}}\cdot \frac{2m\alpha^{t-1}}{\tnk},
 \end{eqnarray}
 where the last inequality is due to Claim~\ref{clm:sum_min_d_cliques} and the setting of $\twt_0$ and $\gamma$.
 The claim follows by the assumption that $\tm\geq m/2$ and the setting of $s_{t+1}$ in  Step~\ref{step:twt_t}.
\end{proof}


\begin{claim}\label{clm:st-dRt}
Let $s_1,\dots,s_k$ be as defined in Step~\ref{step:twt_t} of Algorithm \Approx,
and let $\twt_0$ be as defined in Step~\ref{step:twt_0}.
Then
\[
\frac{s_1 \cldots s_k}{d(\mR_0) \cldots d(\mR_{k-1})}
  = \frac{\tau_k}{(1-\gamma)^{k-1}\cdot \twt_0}\cdot \frac{3\ln(2/\beta)}{\gamma^2}\;.
\]
\end{claim}
\begin{proof}
We prove by induction on $j \in [0,k-1]$ that
\begin{equation}
s_{k-j} \cldots s_k = d(\mR_{k-j-1}) \cldots d(\mR_{k-1}) \cdot \frac{\tau_k}{(1-\gamma)^j\twt_{k-j-1}}\cdot \frac{3\ln(2/\beta)}{\gamma^2}\;,
\end{equation}
and the claim follows by setting $j= k-1$.

For the base of the induction, $j=0$, observe that $s_{k-j} \cldots s_k = s_k$,
and the claim follows directly from the setting of $s_k$.
For the induction step, assume the claim holds for $j \geq 0$ and we prove it for $j+1$.
By the induction hypothesis,
\begin{equation}
s_{k-(j+1)}\cdot s_{k-j} \cldots s_k
 = s_{k-(j+1)} \cdot d(\mR_{k-j-1}) \cldots d(\mR_{k-1}) \cdot \frac{\tau_k}{(1-\gamma)^j\twt_{k-j-1}}\cdot \frac{3\ln(2/\beta)}{\gamma^2}\;.
\end{equation}
The induction step follows
by the setting of $\twt_{k-j-1} = (1-\gamma)\cdot \frac{\twt_{k-j-2}}{d(\mR_{k-j-2})}\cdot s_{k-j-1}$
(in Step~\ref{step:twt_t}).
\end{proof}

We are now ready to prove Theorem~\ref{thm:approx-weight}.

\begin{proof}[Proof of Theorem~\ref{thm:approx-weight}]
We first prove that conditioned on $\tm \geq m/2$,
 with 
probability at least $1-2\delta/3$,
the algorithm does not abort at any step of the invocation, and then continue to prove that its output is as desired.


By
Item~\ref{item:Approx-degrees} of
Claim~\ref{clm:all-R-typical}   with probability at least $1-\delta/3$, for every $t\in [k]$ and $j\in[t,k-1]$,
\confeqn{
	d(\mC_j(\mR_{t})) \leq  \frac{\tnote{(k-t-1)^2}}{\beta} \cdot \frac{d(\mC_j(\mR_{t-1}))}{d(\mR_{t-1})}\cdot s_{t}.
}
Hence, with probability at least 
$1-\delta/3$
the conditions for Claim~\ref{clm:nk_d(R)} hold, and we get that for every $t\in [k]$,
$s_{t+1} \leq \stthres$, so that the algorithm does not abort in Step~\ref{step:abort_st}.

In order to upper bound $|\mR_k|$, we first observe that $|\mR_k| = c_k(\mR_k)$.
By Markov's inequality, $c_k(\mR_1) \leq \frac{k}{\beta}\cdot n_k \cdot \frac{s_1}{n}$ with probability at least $1-\beta$.
By repeated applications of Item~\ref{item:sampleset-ck} in Claim~\ref{clm:sampleset}
we get that with probability at least $1-k\beta=1-\delta/3$,
\begin{equation}
c_k(\mR_k) \leq \frac{(k!)^2}{\beta^k} \cdot n_k \cdot \frac{s_1\cldots s_k}{n\cdot d(\mR_1)\cldots d(\mR_{k-1})}\;.
\end{equation}
By recalling that $d(\mR_0) = n$, $\twt_0 = (1-\eps/2)\tnk$, applying Claim~\ref{clm:st-dRt}, and using $|\mR_k| = c_k(\mR_k)$, we get that with probability at least $1-k\beta = 1-\delta/3$,
\begin{equation}\label{eqn:Rk}
|\mR_k|  \leq \frac{(k!)^2}{\beta^k} \cdot \frac{n_k\cdot \tau_k}{\tnk}\cdot  \frac{12\ln(2/\beta)}{\gamma^2}\;,
\end{equation}
so that the algorithm does not abort at Step~\ref{step:abort_Rk}.
Therefore, with probability at least $1-2\delta/3$, the algorithm does not abort at
any step of the algorithm.
We henceforth condition on these events.

Consider the case that $\mA$ is $(\eps, \vec{\tau})$-good and that
$\tnk \leq n_k$.
By
Item~\ref{item:Approx-weight1} of
Claim~\ref{clm:all-R-typical},
with probability at least $1-\delta/3$, every sample $\mR_t$
satisfies $\wtP(\mR_{t}) \in (1\pm \gamma) \cdot \frac{\wt(\mR_{t-1})}{d(\mR_{t-1})}\cdot s_{t}$.
Conditioned on this  we get that
\begin{equation}
\wtP(\mR_{k}) \in (1\pm \gamma)^k \cdot \frac{\wtP(V) \cdot s_1 \cldots s_k}{n \cdot d(\mR_{1}) \cldots d(\mR_{k-1})}\;,
\end{equation}
and therefore the value $\hnk$
 computed in Step~\ref{step:hnk} is a $(1\pm\gamma)^{k}\in (1\pm \eps/2)$ approximation of $\wtP(V)$.
 By Fact~\ref{fact:wtP-ub},
 we have that $\hnk  < (1+\eps)n_k$.
 Since (by the premise of the theorem)
$\mA$ is \pgood,
we also have that $\hnk\geq (1- \eps)n_k$. Therefore, if $\tm \geq m/2$, $\tnk\leq n_k$ and $\mA$ is $(\eps, \vec{\tau})$-good then
with probability at least $1-\delta$, the algorithm does not abort and returns a value $\hnk \in (1\pm \eps)n_k$.

If $\tnk > n_k$, then it is no longer necessarily true that
the weights of the samples are as stated above.
However,
 we can apply
 Item~\ref{item:Approx-weight2} of Claim~\ref{clm:all-R-typical}, and by the setting of the $\twt_t$'s and $\hnk$ get that with probability at least $1-\delta/3$, $\hnk \leq (1+\gamma)^k \tnk$, which by the setting of $\gamma$ is at most $(1+\eps)\tnk$.

We now turn to the running time (the argument for the query complexity is identical).
The running time is upper bounded by
\begin{equation}\label{eqn:time-ub1}
O\left(s_1 + \sum_{t=1}^{k-1}(|\mR_t| + t\cdot s_{t+1}) + |\mR_k|\cdot k^2\cdot \rho(\op)\right)\;.
\end{equation}
For $t<k$ we simply upper bound $|\mR_t|$ by $s_t$, so that the expression in Equation~\eqref{eqn:time-ub1} is upper bounded by
\begin{align}\label{eqn:time-ub2}
&O\left(s_1 + \sum_{t=2}^{k}(t\cdot s_{t}) + |\mR_k|\cdot k^2\cdot \rho(\op)\right)
=
\\&O\left(\frac{n \tau_1}{\tnk} + \frac{k^{3k}}{\delta^{2k}}\cdot \frac{\tm}{\tnk}\cdot
\sum_{t=2}^{k-1} (\alpha^{t-2}\cdot \tau_t)
+  \frac{k^{3k}}{\delta^{k}}\cdot   \frac{n_k}{\tnk}\cdot  \tau_k \cdot \rho(\op)
\right)\cdot \frac{3^k\cdot \ln(k/\delta)}{\eps^2}
\;.
\end{align}
The upper bound on the query complexity is the same, replacing the term $\rho(\op)$ with $q(\op)$.
\end{proof}

\newcommand{\WNt}{W^{\mN}_t}

\newcommand{\ost}{\frac{2\tm\alpha^{t-1}\cdot \tauu_{t+1}}{\tnk}\cdot \frac{12\ln(1/\beta)}{\beta^k \cdot \gamma^3}}

\def\IsActivProcFull{
\begin{figure}[htb!]
	\fbox{
		\begin{minipage}{0.9\textwidth}
			{\bf \IsActive}($i,\vI,k,\alpha,\eps,\delta,\tnk,\tm, \vec{\tau}$) \label{alg:is-active}
			\smallskip
			\begin{compactenum}
				\item For $\ell =1$ to $q = 12\ln(n^k/\delta)$ do:
				\begin{compactenum}
					\item Let $\mR_i=\big\{\vI\big\}$, $\twt_i = (1-\eps/2)\tau_i$, 
					$\beta=1/(6k)$,  and $\gamma=\eps/(8k \cdot k!)$. \label{step:setbeta}
					\item For $t=i$ to $k-1$ do:
					\begin{compactenum}
						\item Compute $d(\mR_t)$. \label{step:is_act_d}
						\item For $t>i$ set 
						$\twt_t =(1-\gamma)\frac{\twt_{t-1} \cdot s_t}{d(\mR_{t-1})}$ and
						$s_{t+1}= \setst$.
						\label{step:tnk_mRt}
						\item If $s_{t+1}> \ost$, then set $\chi_\ell = 0$ and continue to next $\ell$.
						\label{step:costly}
						\item Invoke
						\SampleSet$(t,\mR_t, s_{t+1})$ 
						and let $\mR_{t+1}$ be the returned multiset.
					\end{compactenum}
					\item Set $\hck(\vI)= \frac{d(\mR_i)\cdot \ldots\cdot d(\mR_{k-1})}{s_{i+1} \cdot\ldots\cdot  s_{k}}
					\cdot |\mR_k|$.\label{step:hnkI}
					\item If $\hck(\vI) \leq \tau_i/4$ then $\chi_\ell = 1$, otherwise, $\chi_\ell=0$. \label{step:st_thres}
					\label{step:return}
				\end{compactenum}
				\item if $\sum_{\ell=1}^q \chi_\ell \geq q/2$ then  \textbf{return} {\sf Active}. Otherwise, \textbf{return} {\sf  Non-Active}.
			\end{compactenum}
			\vspace{3pt}
		\end{minipage}
	}
\end{figure}
}

\ifnum\conf=1

\section{The IsActive procedure}\label{sec:isactive-full}

In this section we give the full details for the procedure \IsActive, which implements
an oracle for a set $\mA$ of active ordered cliques. We prove
that with high probability (over the randomness of the procedure) $\mA$ is \pgood\ for a vector $\vec{\tau}$ that will be set subsequently.
Recall that the set of active cliques was defined so as to ensure that
no ordered $t$-clique $\vT$ for which $U(\vT)$ participates in too many $k$-cliques is assigned any $k$-clique. We refer to cliques that participates in too many $k$-cliques  as sociable cliques as defined in the next subsection.

\else
\section{Implementing an oracle for good subset $\mA$}\label{sec:isactive-full}

In this section we describe a (randomized) procedure named \IsActive, that implements an oracle for a subset $\mA$, where with
high probability, $\mA$ is \pgood\ for an appropriate setting of $\vec{\tau}$.
\dnote{Recalling the discussion in Section~\ref{subsec:sociable-over}, the}
procedure aims at determining whether a given ordered $i$-clique $\vI$ is sociable. That is,
whether $c_k(\vI)$ is larger than  $\tau_i$, in which case it is not included in $\mA$.
This ensures that the first item in Definition~\ref{def:goodCPC} holds, since  $\wtP(\vI) \leq c_k(\vI)\leq \tau_i$.
We also allow the procedure \IsActive\ to abort before achieving an estimate of $c_k(\vI)$.
This early exit (with an output of \textsf{Non-Active})
addresses the case that $\vI$ is costly
(as was informally discussed in Section~\ref{subsec:verify-costly}).

In the following subsections we give formal definitions of the notions of sociable cliques (Section~\ref{subsec:sociable}) and costly  cliques (Section~\ref{subsec:costly}).
The description of the procedure and its analysis are then given in Section~\ref{subsec:isactive}.

\fi

\subsection{Sociable cliques}\label{subsec:sociable}

In this subsection we define ``sociable'' and ``non-sociable'' (ordered and unordered) cliques. That is, we define certain 
thresholds on the number of $k$-cliques that a $t$-clique participates in, according to which we decide whether the $t$-clique is ``sociable'' or ``non-sociable''.
We then give a sufficient condition for a 
subset $\mA$ of ordered cliques
to be \pgood\ (recall Definition~\ref{def:goodCPC}) based on these thresholds.


\newcommand{\taut}{ \frac{k^{4k}}{\beta^k\cdot \gamma^2}\cdot \alpha^{k-t}}


In what follows we  set $\gamma = \eps/(8k \cdot k!)$ and \tnote{$\beta = 1/(6k)$}.
\begin{definition}[Sociability thresholds] \label{def:tau_i}
For each $t \in [2,k-1]$, we set $$\tauu_t=\taut\;.$$
For $t=1$ and an estimate $\tnk$ of $n_k$,
 we set $\tauu_1 = \frac{k^{4k}}{\gamma^2}\cdot \min\left\{\alpha^{k-1}, \tnk^{\frac{k-1}{k}}   \right\}$,
and for $t=k$ we set $\tauu_k=1$.
We refer to $\tauu_t$ as the ``$t$-sociability threshold".
For $t\in [k-1]$ we set 
$\taul_t = \beta^k \tauu_t/(4(k!)^2)$ 
and $\taul_k = \tauu_k$.
\end{definition}

\begin{definition}[Sociable cliques] \label{def:sociable}
Let $t \in [k].$
	We say that a $t$-clique $T$ is 
{\sf sociable} (with respect to $\tnk$),
if $c_k(T)> \tauu_t$.
 We say that a $t$-clique $T$  
 is {\sf non-sociable} (with respect to $\tnk$), if
 $c_k(T) \leq \taul_t$.
We say that an ordered clique $\vT$ is sociable (non-sociable) if $U(\vT)$ is sociable (respectively, non-sociable).
\end{definition}
Observe that every (ordered) $k$-clique is non-sociable.
We also note that the special setting of $\tauu_1$ (the threshold for vertices) is
due to the need to deal separately with the case that $n_k \leq \alpha^k$ and the case that $n_k > \alpha^k$.

We shall prove that if $\mA$ contains all non-sociable ordered cliques, then
$\wtP(V)$ is not much smaller than $n_k$.
We actually prove the claim for  a generalization of the weight function introduced in Definition~\ref{def:wgt}.
Given an ordered $i$-clique $\vI$, we may  restrict our attention to $k$-cliques in $\mC_k(\vI)$
and assign them to ordered $t$-cliques in $\mOC_t(\vI)$ for $t > i$
(independently of whether $\vI$ itself belongs to $\mA$ or not).

\begin{definition}[Fully-active cliques with respect to $\vI$] \label{def:fully-active-I}
	Let  $\mA $ be a subset of  ordered cliques and let $\vI$ be an ordered $i$-clique.
	Consider any ordered $t$-clique $\vT$ that is an extension of $\vI$, that is, $\vT\in \mOC_t(\vI)$.
 We say that $\vT$ is {\sf fully active} with respect to $\mA$ and $\vI$,
	if 	all of its prefixes that are extensions of $\vI$
	belong to $\mA$. That is, $\vT_{\leq j} \in \mA$ for every $j \in [i+1,t]$.
	For $t\in [i+1,k]$, we denote the subset of ordered $t$-cliques that are fully active with respect to $\mA$ and $\vI$ by $\mFAI_t$.
\end{definition}

\begin{definition}[Assignment and weight] \label{def:wgt-I}
	Let $\mA$ be a subset of ordered cliques and $\vI$ be an ordered $i$-clique (that may not belong to $\mA$).

	For each $k$-clique $C$ that contains $U(\vI)$, if $U^{-1}(C) \cap \mFAI_k \neq \emptyset$, then $C$ is
	{\sf assigned}  (with respect to $\mA$ and $\vI$) to the first ordered $k$-clique $\vC \in U^{-1}(C) \cap \mFAI_k$,
	and to each ordered $t$-clique $\vC_{\leq t}$ for $t \in [i,k-1]$.
	Otherwise ($U^{-1}(C) \cap \mFA_k =\emptyset$), $C$ is {\sf unassigned}. That is, for each $k$-clique $C\in \mC_k(\vI)$, if there is some ordered $k$-clique of $C$ that is an extension of $\vI$ and which is fully-active with respect to $\mA$ and $\vI$, then $\vI$ is assigned to the first such ordered $k$-clique. Otherwise, $C$ is unassigned.
	
	\label{par:wtI}For each ordered $t$-clique $\vT$, we let $\wtAI(\vT)$  denote the number of $k$-cliques that are assigned to $\vT$ with respect to $\vI$, and we refer to $\wtAI(\vT)$ as the {\sf weight of $\vT$} with respect to $\mA$ and $\vI$.
\end{definition}
\renewcommand{\wtP}{\wtAI}

%
Observe that Definition~\ref{def:wgt} is a special case of Definition~\ref{def:wgt-I} when we
take $\vI$ to be the single ordered $0$-clique, which we denote by $\lambda$.

\newcommand{\piAI}{\pi^{\mA,\vI}}
\begin{lemma}\label{lem:almost-good-P}
	Let $\mA = \bigcup_{t=1}^k \mA_t$ be a subset of ordered cliques
	such that for every $t \in [k]$, the
	subset $\mA_t$ contains all non-sociable ordered $t$-cliques.
	For any ordered $i$-clique $\vI$ such that $i>0$,
\begin{equation}\label{eqn:wtPmOCkI}
	\wtP(\mOC_k(\vI))  \geq (1-k\cdot \gamma)c_k(\vI)\;,
\end{equation}
and for $i=0$ Equation~\eqref{eqn:wtPmOCkI} holds conditioned on $\tnk > (2\alpha)^{k}$
or $\tnk \geq n_k/4$.
\end{lemma}
Observe that for the case that $i=0$ and $I=\lambda$, it holds that $c_k(\vI) = n_k$ and $\wtAI=\omega^{\mA}$, so
Equation~\eqref{eqn:wtPmOCkI} is equivalent to
\begin{equation}
\omega^{\mA}(V) \geq  (1-k\cdot \gamma)n_k\;.
\end{equation}

\begin{proof}
	In order to prove Lemma~\ref{lem:almost-good-P} we consider an  \emph{iterative potential assignment process},
	and use this process to lower bound $\wtAI(\mOC_k(\vI))$.
	In each iteration $t$ we define a mapping $\piAI_t:\mC_k(\vI) \rightarrow \mFAI_t$ that
	potentially assigns $k$-cliques to ordered $t$-cliques that are fully-active  (with respect to $\mA$ and $\vI$). We let $\mPAI_t$ be the subset of $k$-cliques that are potentially assigned at the end of the $t\th$ iteration, where $|\mPAI_i|=c_k(\vI)$ and for every $t\in [i,k-1]$, $\mPAI_{t+1} \subseteq \mPAI_t$.	We prove by induction that almost all of the $k$-cliques  in $\mC_k(\vI)$ ``survive" this potential assignment process, so that
	for every $t\geq i$,
	\[
|\mPAI_t |\geq (1- t\cdot \gamma) \cdot c_k(\vI).
	\numberthis
	\]
	
	At the first iteration, $t=i$, and all the $k$-cliques in $\mC_k(\vI)$ are potentially-assigned to $\vI$. Hence, $|\mPAI_i|=c_k(\vI)$ and the inequality holds for $t=i$.
	We now assume that the inequality holds for  $j \in [i,t]$ and prove it for $t+1$.
	
	Let $C$ be a $k$-clique in $\mPAI_t$, let $\vT = \piAI_t(C)$ be the ordered $t$-clique that $C$ is potentially assigned to, and let
$\mOC_{t+1}(\vT,C) \subset \mOC_{t+1}(\vT)$ be the set of single-vertex
extensions $\vT'$ of $\vT$ such that $U(\vT') \subseteq C$. That is,
$\mOC_{t+1}(\vT,C)= \{(\vT, u) \mid u \in C \setminus U(\vT) \}$.
If $\mOC_{t+1}(\vT,C) \cap \mFAI_{t+1} \neq \emptyset$,  then $C$ is potentially assigned to the first ordered $(t+1)$-clique in $\mOC_{t+1}(\vT,C)  \cap \mFAI_{t+1}$, and otherwise it is not potentially assigned (nor assigned) to any ordered $(t+1)$-clique.
	Therefore, a $k$-clique $C \in \mPAI_t$ is not in $\mPAI_{t+1}$ if all of the ordered $(t+1)$-cliques  $(\vT,u)$ for $u \in C \setminus U(\vT)$ are not $(t+1)$-fully active. 	We shall upper-bound the number of such  $k$-cliques.
\tnote{
For every ordered $t$-clique $\vT$, consider the following auxiliary subgraph $W_{t,\vT}$.
For each $k$-clique $C\in \mC_k(\vT)$ there are $k-t$ nodes
in $W_{t,\vT}$. Each of these $k-t$ nodes corresponds to one of the ordered $(t+1)$-cliques in
$\mOC_{t+1}(\vT,C)$.
There is an edge between two nodes $(\vT,u)$ and $(\vT,u')$ in $W_{t,\vT}$ if (and only if)
there is an edge between $u$ and $u'$ in $G$.

	Let $W_t$ be the (disjoint) union of all subgraphs $W_{t,\vT}$ over all the ordered $t$-cliques $\vT$ that are potentially assigned a $k$-clique. We say that a node in $W_t$ is {\em non-active\/} if its corresponding ordered $(t+1)$-clique is not in $\mA_{t+1}$.
We denote the subgraph induced by the set of non-active nodes in $W_t$ by $\WNt$.
	Observe that there is a one-to-one correspondence between $k$-cliques in $\mC_k(\vT)$ and $(k-t)$-cliques in $W_t$. Furthermore, the set of $k$-cliques of $\mC_k(\vT)$ that are not in $\mPAI_{t+1}$ can be classified into two types. The first type is simply the $k$-cliques that are not in $\mPAI_{t}$ and the second type is $k$-cliques that are in $\mPAI_{t}$ but are not in $\mPAI_{t+1}$.	Let $C$ by a $k$-clique of the second type. Then $C$ is assigned to
	 some $t$-clique $\vT$, but all of the $(k-t)$ nodes in $W_{t,\vT}$ that correspond to the ordered $(t+1)$-cliques in $\mOC_{t+1}(\vT,C)$  are non-active so that $C$ cannot be further assigned to any ordered $(t+1)$-clique.
	Therefore, the $k$-cliques of the second type correspond to
	$(k-t)$-cliques in $W_t$ that all of their nodes are non-active.
	That is, they correspond to  $(k-t)$-cliques in $\WNt$. Hence, we shall want to bound the number of $(k-t)$-cliques in $\WNt$.
	
	By the assumption that $\mA$ contains all non-sociable $t$-cliques, it holds that any ordered $(t+1)$-clique that is non-active participates in at least $\tau_{t+1}^{L}$ $k$-cliques in $\mC_k(\vT)$ (which in turn correspond to $\tau_{t+1}^{L}$ $(k-t)$-cliques in $W_t$). Hence, we can bound the number of nodes in the subgraph $\WNt$ (denoted $n_1(\WNt)$) as follows.

	\[ n_1(\WNt) < \frac{ (k-t)\cdot n_{k-t}(W_t)}{\tau^L_{t+1}}=
	\frac{(k-t)\cdot c_k(\vI)}{\tau^L_{t+1}}\;,
	\numberthis
	\]
	where the $(k-t)$ factor is due to the fact that every $(k-t)$-clique  can be counted from its $(k-t)$ nodes.
}
	We first continue the proof for the case that $i\neq 0$ or that $i=0$ and $\tnk > \alpha^{k}$
By	 applying Corollary~\ref{cor:nk_vs_nt} to the auxiliary graph $\WNt$ (for $k=k-t$ and $t=1$),
	\[
	n_{k-t}(\WNt) \leq
	\frac{n_1(\WNt)}{(k-t)!}\cdot  (\alpha(\WNt))^{k-t-1} \leq
	\frac{(k-t)\cdot c_k(\vI)}{(k-t)!\cdot \tau^L_{t+1}}\cdot \cdot  (\alpha(\WNt))^{k-t-1}.
	\numberthis \label{eqn:n_k-t(W)_alpha}
	\]
	Observe that each graph $W_{t,\vT}$ corresponds to a subgraph of $G$, so that for every $\vT$,
	$\alpha(W_{t,\vT}) \leq \alpha(G)\leq \alpha$.
Since $W_t$ is a disjoint union of graphs with arboricity at most $\alpha$, and since
$\WNt$ is a subgraph of $W_t$, it follows that 	$\alpha(\WNt) \leq \alpha$.
	Plugging this into Equation~\eqref{eqn:n_k-t(W)_alpha} together with the setting of $\tau^L_{t+1}$
from Definition~\ref{def:tau_i} (where for $t=i=0$ we use the premise of this case by which $\tnk > \alpha^{k}$),
we get that
	\[
	n_{k-t}(\WNt) < \gamma \cdot c_k(\vI)\;. \numberthis
	\]

	We now consider the case that $i=0$ (so that $c_k(\vI) = n_k$) and $\tnk \leq \alpha^{k}$ .
	Clearly, for any graph $W$ and $k$, $n_k(W) \leq (n_1(W))^k$. Hence, by the setting of 
$\taul_1$
for this case, we get that if
$\tnk \geq n_k/4$, 
then
	\[n_{k-t}(\WNt) < 	\big(n_{1}(\WNt)\big)^k < \gamma \cdot n_k =	\gamma \cdot c_k(\vI)\;. \numberthis \label{eqn:n_k-t(W)}
	\]

	By Equations~\eqref{eqn:n_k-t(W)_alpha} and~\eqref{eqn:n_k-t(W)} and by  the induction hypothesis,  it follows for both cases that
	\[
	|\mPAI_{t+1}| > 	|\mPAI_t| - n_{k-t}(\WNt) > (1-t\cdot \gamma)c_k(\vI) \;,
	\]
	thus establishing the induction step.
	Hence,
	\[
		|\mPAI_k| > 	  (1-(k-i) \cdot \gamma)c_k(\vI)  \;.
	\]
	
Since for every $k$-clique $C$ in $\mPAI_k$ there is at least one fully-active ordered  $k$-clique in $U^{-1}(C)$, it follows from Definition~\ref{def:wgt} that $C$ is assigned (with respect to $\mA$ and $\mI$) to some ordered $k$-clique in $\mOC_k(\vI)$, so that
	\[ \wtP(\mOC_k(\vI))>  		|\mPAI_k|	 > (1-(k-i) \cdot \gamma)n_k(I)  > (1-k \cdot \gamma)c_k(\vI)\;,
	\]
	as claimed.
\end{proof}

\subsection{Costly cliques}\label{subsec:costly}
Ideally we would have liked the procedure \IsActive\ to distinguish between non-sociable and sociable cliques (where for the former we would like it to return that they are active and for the latter that they are not active).
However, in order to bound the complexity of \IsActive, we shall actually allow it to decide that some  cliques are not active, even though they are  non-sociable. This is since for some cliques it might be too costly (in terms of running time) to determine whether they are sociable. Therefore, if the procedure identifies a clique as costly it  returns that it is not active, and we prove that this only amounts to a small loss in the estimation.
 To formalize this, we introduce the next definition (which refers to unordered
 cliques).
Recall that by Claim~\ref{clm:sum_min_d_cliques},  $d(\mC_j) \leq 2m\cdot\alpha^{j-1}$, for every $j$.
This bound was an important ingredient in the analysis of the complexity of \Approx\ (see Claim~\ref{clm:nk_d(R)}
and Theorem~\ref{thm:approx-weight}). 

\begin{definition}[Costly cliques] \label{def:costly}
	For $i\in [k-1]$ and $j \in [i,k-1]$, an $i$-clique $I$ is {\sf $j$-costly}
with respect to an estimate $\tnk$ of $n_k$ if
\[\frac{d(\mC_j(I))}{c_k(I)} > \frac{2m\alpha^{j-1}}{\gamma \cdot \tnk} \;.\]	
It is {\sf costly} if it is $j$-costly for some $j$.
\end{definition}

\begin{claim}\label{clm:rel-cost}
For an estimate $\tnk$, the number of $k$-cliques that contain some $i$-clique that is costly  with respect to $\tnk$ (for any $i\in [k-1]$) is at most
$2^{k+1}\gamma \tnk$.
\end{claim}
\begin{proof}
Fix $i \in [k-1]$ and $j\in [i,k-1]$.
For every $j$-costly $i$-clique $I$,
\begin{equation}
 c_k(I) < \frac{d(\mC_j(I)) \cdot \gamma \cdot \tnk}{2m \cdot \alpha^{j-1}}\; .
\end{equation}
Summing over both sides of the equation,  this gives
\begin{align}
 \sum_{\substack{I \in \mC_i \\ \text{$I$ is $j$-costly}}}c_k(I)
 &< \sum_{\substack{I \in \mC_i \\ \text{$I$ is $j$-costly}}} \frac{d(\mC_j(I)) \cdot \gamma \cdot \tnk}{2m \cdot \alpha^{j-1}} \\
  &\leq  \frac{\gamma \cdot \tnk}{2m \cdot \alpha^{j-1}}  \sum_{\substack{I \in \mC_i }} d(\mC_j(I))
  =  \binom{j}{i} \cdot \frac{\gamma \cdot \tnk}{2m \cdot \alpha^{j-1}}   \sum_{\substack{J \in \mC_j }} d(J)
  \\ &\leq \binom{j}{i} \cdot \frac{\gamma \cdot \tnk}{2m \cdot \alpha^{j-1}}  \cdot 2m \cdot \alpha^{j-1}  =  \binom{j}{i} \cdot \gamma \cdot \tnk\;,
\end{align}
where the inequality before last  is due to Claim~\ref{clm:sum_min_d_cliques}.
Therefore, if we fix $j$, and sum over all possible $i$'s in $[1,j]$, there are at most $2^j \cdot  \gamma \tnk$
$k$-cliques in which only $j$-costly $i$-cliques participate in.
The claim follows by summing over all possible $j$'s in $[k-1]$.
\end{proof}

By combining Lemma~\ref{lem:almost-good-P} with Claim~\ref{clm:rel-cost} we get the following.
In what follows we say that an ordered $t$-clique $\vT$ is  costly if $U(\vT)$ is  costly.
\begin{claim}\label{clm:sufficient-conditions-for-good-CPC}
Let 
$\mA = \bigcup_{t=1}^k$ be a subset of ordered cliques such that
for every $t \in [k]$, the subset $\mA_t$ contains all non-sociable ordered $t$-cliques 
that are not costly with respect to $\tnk$. 
If $\tnk\geq n_k/4$, then
$\wtP(V) \geq (1-\eps/4)n_k -(\eps/4)\tnk$.
\end{claim}
\begin{proof}
By Lemma~\ref{lem:almost-good-P} (with $\vI=\lambda$), if for every $t \in [k]$, the subset $\mA_t$ contains {\em all\/} non-sociable ordered $t$-cliques,
then $\wtP(V) \geq (1-2k \cdot 2^k \cdot \gamma )n_k$, which, by the setting of $\gamma$ is at least $(1-\eps/4)n_k$.
By the premise of this claim, for every $t \in [k]$ the subset $\mA_t$ might contain only those non-sociable ordered
$t$-cliques that are not costly. As a consequence, $k$-cliques that contain costly $t$-cliques might
not be assigned to any vertex.
However, by Claim~\ref{clm:rel-cost}, the total number of $k$-cliques in which {\em some\/}  costly clique participate in  is at  most $k \cdot 2^k \cdot \gamma \tnk$, and the current claim follows.
\end{proof}

\subsection{The procedure Is-Active}\label{subsec:isactive}

\newcommand{\x}{\frac{3\ln(1/\beta)}{\gamma^2}}

\ifnum\conf=1
We start by presenting \IsActive\ in full detail.
\fi
                  				
\IsActivProcFull

We next prove two claims regarding \IsActive. The first is for the case that the given ordered clique $\vI$ is
non-sociable and not costly, and the second is for the case that $\vI$ is sociable.
\begin{claim}\label{clm:IsActive-nonsociable}
	Let $\tm \geq m/2$ and let
	$\vI$ be a non-sociable  ordered $i$-clique for $i \in [k-1]$ that is not costly with respect to $\tnk$.
Consider an invocation of \IsActive\ on $\vI$ with $\vec{\tau}$ set to $\vtauu$.
With probability at least $1-\delta/n^k$, \IsActive\ returns \textsf{Active}.
\end{claim}
\begin{proof}
Consider a single iteration of the for loop on $\ell$. We shall show that
$\chi_\ell = 1$ with probability at least $2/3$. The claim then follows by applying a multiplicative Chernoff bound
on the $\chi_\ell$'s.
By
Item~\ref{item:sampleset-degrees} of
 Claim~\ref{clm:sampleset}
for $t \in [i,k-1]$,
with probability at least \tnote{$1-\beta$},
  for every $j \in [t+1,k-1]$,
\begin{equation}
d(\mC_j(\mR_{t+1})) \leq   \frac{(k-t)^2}{\beta}\cdot \frac{d(\mC_j(\mR_{t}))}{d(\mR_{t})}\cdot s_{t+1}\;. \numberthis \label{eqn:d(C_j)}
\end{equation}
Hence, Equation~\eqref{eqn:d(C_j)} holds with probability at least $1-k\cdot \beta$ for every $t\in[i,k-1]$ and $j \in [t+1,k-1]$.
We henceforth condition on this event and prove that it follows that
for every $t \in [i,k-1]$,
$s_{t+1}\leq \overline{s}_{t+1}$, where $\overline{s}_{t+1}$ is the threshold set in Step~\ref{step:st_thres},
\begin{equation}
\overline{s}_{t+1}=\ost \;.
\numberthis\label{eqn:bound_st}
\end{equation}
By Definition~\ref{def:costly}, 
since $\vI$ is non-sociable and not costly with respect to $\tnk$,
\begin{equation}
d(\mR_i) =d(\vI)=d(\mC_i(\vI))\leq \frac{2m \alpha^{i-1} \cdot \taul_i}{\gamma \tnk}\;. \numberthis \label{eqn:dRi}
\end{equation}
By the setting of $s_{i+1}$ and $\twt_{i}$ in Step~\ref{step:tnk_mRt}, the setting of  $\vtauu$ and $\vtaul$ in Definition~\ref{def:tau_i}, and the premise on $\tm$,  
\begin{align}
s_{i+1}&=\frac{d(\mR_i)\cdot \tauu_{i+1} }{(1-\gamma)\twt_{i}} \cdot \x \leq
\frac{2m \alpha^{i-1}\cdot (\tauu_i/(2(k!)^2))\cdot \tauu_{i+1}}{(1-\gamma)\cdot (1-\eps/2)\cdot \tauu_i\cdot \gamma \cdot \tnk} \cdot \x  \\
&\leq  \frac{2m\alpha^{i-1}\cdot \tauu_{i+1} }{\tnk}\cdot \frac{ 6\ln(1/\beta)}{ \gamma^3\cdot (k!)^2} < \overline{s}_{i+1}
\;.
\end{align}
Hence, the inequality holds for $t=i$.

By the conditioning that the inequality in Equation~(\ref{eqn:d(C_j)}) holds for every pair $t \in [i+1, k-1]$ and $j \in [t, k-1]$, we have that for $t\in[i+1,k-1]$
\begin{equation}
d(\mR_{t})=d\big(\mC_t(\mR_t)\big)<\frac{(k-t)^2}{\beta}\cdot \frac{d(\mC_{t}(\mR_{t-1}))\cdot s_{t}}{d(\mR_{t-1})}\;.
\end{equation}
By the setting $\twt_{t}=\frac{(1-\gamma)\twt_{t-1}\cdot s_{t}}{d(\mR_{t-1})}$ for $t\in[i+1,k-1]$ in Step~\ref{step:tnk_mRt},
\begin{align}
s_{t+1}&=\frac{d(\mR_t)\cdot \tauu_{t+1}}{\twt_{t}}\cdot \x<
   \frac{(k-t)^2}{\beta}\cdot\frac{d(\mC_t(\mR_{t-1}))\cdot \tauu_{t+1}}{\twt_{t-1}}\cdot \x\\
&< \ldots < \frac{(k-t)^2\cdot \ldots\cdot (k-i-1)^2}{\beta^{t-i+1}}\cdot \frac{d(\mC_t(\mR_i))\cdot \tauu_{t+1}}{\twt_{i}}\cdot \x\\
&< \left(\frac{(k-t)!}{(k-i)!}\right)^2\cdot \frac{1}{\beta^{t-i+1}}\cdot  \frac{2m \alpha^{t-1} \cdot \taul_i}{\gamma \cdot \tnk}
             \cdot \frac{\tauu_{t+1}}{(1-\eps/2)\tauu_i}\cdot \x\\
&< \frac{2m\alpha^{t-1}\cdot \tauu_{t+1}}{\tnk}\cdot \frac{6\ln(1/\beta)}{\beta^k \gamma^3 }
< \overline{s}_{t+1}\;.
\end{align}
as claimed.
 It follows that with probability at least $1-k \cdot \beta$,
the procedure does not
 set $\chi_\ell=0$
 in Step~\ref{step:costly} for any $t\in[i,k-1]$. We henceforth condition on this event.

It remains to
bound the probability that $\chi_\ell$ is set to $0$
in Step~\ref{step:return} for $t=k-1$.
By 
Item~\ref{item:sampleset-ck} of Claim~\ref{clm:sampleset}
and the union bound, with probability at least $1-(k-i-1)\beta$, for every $t\in[i,k-1]$,
\begin{equation}
c_k(\mR_{t+1})\leq \frac{(k-t)}{\beta}\cdot \frac{c_k(\mOC_{t+1}(\mR_t))\cdot s_{t+1}}{d(\mR_t)}
  = \frac{(k-t)^2}{\beta}\cdot \frac{c_k(\mR_t)\cdot s_{t+1}}{d(\mR_t)}\;. 
\end{equation}
Hence, with probability at least $1-(k-i-1)\beta$,
\begin{align}
|\mR_k|=c_k(\mR_k)&<  c_k(\mR_{k-1})\cdot \frac{s_{k}}{d(\mR_{k-1})}\\
&< \frac{1^2 \cdot 2^2}{\beta^2}\cdot c_k(\mR_{k-2})\cdot \frac{s_{k-1}\cdot s_k}{d(\mR_{k-2}\cdot \mR_{k-1})}
\\
&< \frac{1^2\cdot2^2 \ldots (k-i-1)^2\cdot (k-i)^2}{\beta^{k-i}} \cdot c_k(\mR_i) \cdot \frac{s_{i+1} \ldots s_{k}}{d(\mR_i)\cdot \ldots\cdot d(\mR_{k-1})}
\\ &< \frac{(k!)^2}{\beta^k} \cdot c_k(\mR_i) \cdot \frac{s_{i+1} \cdot \ldots \cdot s_{k}}{d(\mR_i)\cdot \ldots\cdot d(\mR_{k-1})}
\;.
\end{align}
Recall that if $\vI$ is non-sociable, then $c_k(\vI)=c_k(\mR_i) \leq \taul_i = \beta^k\tauu_i/(4(k!)^2)$.
Therefore, if $\vI$ is non-sociable, then with probability at least $1-(k-i-1)\beta$ it holds that
$\hck(\vI) \leq \tauu_i/4$ (where $\hck(\vI)$ is as set in Step~\ref{step:hnkI}).
Therefore, with probability at least $1-2k\cdot \beta>2/3$, $\chi_\ell$ is set to $1$ in Step~\ref{step:return}.
\end{proof}

\newcommand{\blah}{\frac{12\ln(1/\beta)}{\gamma^2}}
\begin{claim}\label{clm:IsActive-sociable}
	Let $\tm \geq m/2$  and let $\vI$ be an ordered sociable $i$-clique for $i\in [k-1]$. Consider an invocation of \IsActive\ on $\vI$ with $\vec{\tau}$ set to $\vtauu$.
With probability at least $1-\delta/n^k$, \IsActive\ returns \textsf{Non-Active}. 
\end{claim}
\begin{proof}
Consider a single iteration of the for loop on $\ell$. We shall show that
$\chi_\ell = 0$ with probability at least $2/3$. The claim then follows by applying a multiplicative Chernoff bound
on the $\chi_\ell$'s.
	If for any $t\in[i,k-1]$, the procedure 
sets $\chi_\ell = 0$ in Step~\ref{step:costly}, then we are done.
Therefore we prove that with probability a least $2/3$,
 if the procedure reaches Step~\ref{step:return}, then it
 sets $\chi_\ell = 0$.

 We next define a weight function over ordered cliques (based on cliques that $\vI$ participates in).
For each $t\in [k]$, let $\mA^*_t \subset \mOC_t$ consist of all non-sociable ordered $t$-cliques, and let
\tnote{$\mA^* = \bigcup_{t=1}^k \mA^*_t$}.
 We let  $\wtPI$ be as defined in Definition~\ref{def:wgt-I} with respect to $\vI$ and $A^*$.
 Observe
 that $\wtPI$ assigns bounded weights to ordered $t$-cliques for $t>i$.
 Namely, $\wtPI(\vT) \leq \tau_t$ for every ordered $t$-clique $\vT$ such that $t>i$.
 Also note that by Lemma~\ref{lem:almost-good-P} and the setting of $\gamma$,
 \begin{equation}\label{eqn:wtPII}
 \wtPI(\vI) \geq (1-\eps/2)c_k(\vI)\;.
 \end{equation}

 We next prove that with probability at least $1-k\cdot \beta$,
 for every $t \in \{i,\dots,k-1\}$, 
 $\twt_{t} \leq \wtPI(\mR_{t})$ for $\twt$ as defined in Step~\ref{step:tnk_mRt} and that
 the sample $\mR_{t+1}$ satisfies:
 $\wtPI(\mR_{t+1}) \geq (1-\gamma) \cdot \frac{\wtPI(\mR_{t})}{d(\mR_{t})}\cdot s_{t+1}$.
 We prove the claim by induction. Namely, we prove that these bounds hold for $t=i$
 with probability at least $1-\beta$, and then prove that for each $t>i$, the bounds hold
 with probability at least $1-\beta$ conditioned on them holding for $t-1$.

 For the base of the induction ($t=i)$ we recall that
 since $\vI$ is sociable, $c_k(\vI) \geq \tauu_i$.
 By the definition of $\twt_i$ and Equation~\eqref{eqn:wtPII} we have that
 \begin{equation}
 \twt_i = (1-\eps/2)\tauu_i \leq (1-\eps/2)c_k(\vI) \leq \wtPI(\vI) = \wtPI(\mR_i)\;.
 \end{equation}
We can therefore apply Item~\ref{item:sampleset-weight1} 
in Claim~\ref{clm:sampleset} to $t=i$ and obtain
that $\wtPI(\mR_{i+1}) \geq (1- \gamma) \cdot \frac{\wtPI(\mR_{i})}{d(\mR_{i})}\cdot s_{i+1}$
 with probability at least $1-\beta$.
 The induction step follows from the setting of $\twt_t$, the induction hypothesis and
 an application of 
 Item~\ref{item:sampleset-weight1} in
 Claim~\ref{clm:sampleset}.

It follows that with probability at least $1-k\cdot \beta\geq 2/3$,
\begin{equation}\label{eqn:wtPstar_I}
\wtPI(\mR_k) \geq (1-\gamma)^{k-i} \cdot \wtPI(\vI) \cdot
\frac{s_{i+1}\cdot \ldots \cdot s_k}{ d(\mR_i) \cdot \ldots \cdot d(\mR_{k-1})}\;.
\end{equation}
Clearly, $|\mR_k| \geq \wtPI(\mR_k)$, so by	 Equations~\eqref{eqn:wtPII} and~\eqref{eqn:wtPstar_I},
the setting of $\hck(\vI)$ in Step~\ref{step:st_thres}, the setting of $\gamma$, and the assumption that $\eps \leq 1/2$, $\hck(\vI) > c_k(\vI)/4 \geq \tauu_i/4$,
	and the procedure sets $\chi_\ell = 0$. 
\end{proof}


We are now ready to prove our main lemma regarding the correctness and complexity of \IsActive.

\begin{lemma}\label{lem:IsActive}
	Consider running \IsActive\ on all ordered cliques of size $t\in [k-1]$
	with $\vec{\tau}$ set to $\vtauu$ (as defined in Definition~\ref{def:tau_i})
	and let 
	$\mA$ be the subset of ordered cliques on which it returns {\sf Active}.
	If 
	$\tnk \in [n_k/4,n_k]$ and $\tm \geq m/2$,
	then 
	$\mA$ is $(\eps, \vtauu)$-good
	with probability at least $1-\delta$.
	If $\tnk> n_k$, then 
	$\mA$ is $\vtauu$-bounded with probability at least $1-\delta$.
	
	The query complexity and running time of a single invocation of \IsActive\ are
	\[
	O\left(\frac{\tm}{\tnk}\cdot \sum_{t=2}^{k} (\alpha^{t-2}\cdot \tauu_t) \right)\cdot
	\frac{k^{6k} \cdot  \log^2(n/\delta)}{\eps^3}\;.
	\]
	
\end{lemma}
\begin{proof}
	For $\mA$ as defined in the lemma, let $\mA_t = \mA \cap \mOC_t$.
	By Claims~\ref{clm:IsActive-nonsociable} and~\ref{clm:IsActive-sociable},
	together with the union bound (over all $t$-cliques in the graph for every $t\in [k-1]$),
	with probability at least $1-\delta$: (1) for every $t \in [k-1]$, the subset 
	$\mA_t$ contains all non-sociable ordered $t$-cliques	 that are not costly with respect to $\tnk$, and
	(2) for every $t \in [k-1]$, the subset 
	$\mA_t$ does not contain any sociable ordered $t$-clique.
	We condition on both events.
	By (1) and Claim~\ref{clm:sufficient-conditions-for-good-CPC}, if $\tnk\in [n_k/4, n_k]$,
	then $\wtP(V)> (1-\eps/2)n_k$. By (2),	for every $t\in[k-1]$ and an ordered  $t$-clique $T$, $\wtP(T)< \tauu_t$.
	Hence, if $\tnk\in [n_k/4, n_k]$, then
	$\mA$ is $(\eps, \vtauu)$-good,
	and if $\tnk > n_k$, then it is $\vtauu$-bounded.

We now turn to analyze the complexity of the procedure. Consider an invocation of \IsActive\ on some  ordered $i$-clique.
Let $\overline{s}_{t+1}=\ost$ denote the threshold defined in Step~\ref{step:costly} of the procedure.
By Claim~\ref{clm:sampleset}, the query complexity and running time of each invocation of \SampleSet\ with parameters $(t,\mR_t,s_{t+1})$  is $O(|\mR_t|+ t\cdot s_{t+1})$. Also, the query complexity and running time of computing $d(\mR_t)$ in Step~\ref{step:is_act_d} is $O(|\mR_t|)$.
 Since for every $t \in [i+1,k]$, we can bound $|\mR_t| \leq s_t$, it follows that
the query complexity and running time of \IsActive\ are upper bounded by
\begin{equation}
O\left(\log(n^k/\delta)\cdot \sum_{t=i+1}^{k}t \cdot \overline{s}_{t}\right)
\;,
\end{equation}
and the claim follows by the setting of $\beta$ and $\gamma$ in Step~\ref{step:setbeta} of the procedure.
\end{proof}

\section{Finalizing}\label{sec:final}

In this section we prove Theorem~\ref{thm:main} (restated as Corollary~\ref{cor:main}).
We first combine the (oracle-aided) algorithm \Approx\ with calls to \IsActive\ (instead of the oracle).
Since both \Approx\ and \IsActive\ need a constant-factor estimate $\tm$ of $m$, the combined algorithm first obtains such an
estimate. This is done by calling the moments-estimation algorithm of~\cite{ERS17} (for the first moment).
The~\cite{ERS17} algorithm is designed to work for bounded-arboricity graphs.

\begin{figure}[htb!]
	\fbox{
		\begin{minipage}{0.95\textwidth}
			{\bf \Main}($n,k,\alpha,\eps,\delta, \tnk$) \label{alg:main}
			\smallskip
			\begin{compactenum}
				\item Call the~\cite{ERS17} algorithm $s=\Theta(\log(2n^2/\delta))$ times with parameters $n,\alpha$ and $\eps=1/2$ to get $s$ independent estimates of $m$, and let $\tm$ be the median of the returned values.
				\item Set $\vec{\tau}=\vtauu$ as defined in Definition~\ref{def:tau_i}.
                \item For $j = 1$ to $q = \Theta(\log(1/\delta))$ do:
                   \begin{compactenum}
                    \item Invoke \Approx$(n,k,\alpha,\eps,1/6,\tnk,\tm,\vec{\tau},\op)$ where each call to $\op$ on an ordered $i$-clique $\vI$
                     is replaced by an invocation of \IsActive$(i,\vI,k,\alpha,\eps,\delta/4,\tnk,\tm, \vec{\tau})$. (If there is more than
                     one call to $\op$ with the same $\vI$, then the output of the first invocation of \IsActive\ is used.)
                     \item Let $\chi_j$ be the returned value.
				   \end{compactenum}
 				\item Let $\hnk$ be the median value of $\chi_1,\dots,\chi_j$ and return $\hnk$.
			\end{compactenum}
			\vspace{3pt}
		\end{minipage}
	}
\end{figure}



\begin{lemma}\label{lem:Main}
	Consider an invocation of  \Main\ with query access to a graph $G$
	and parameters $n,k,\alpha, \eps, \delta$ and $\tnk$. Then the following holds.
	\begin{itemize}
	\item If $\tnk \in [n_k/4, n_k]$, then \Main\
	returns a value $\hnk$ such that with probability at least $1-\delta$, $\hnk$ is a  $(1\pm\eps)$-approximation of $n_k(G)$.
	\item If $\tnk > n_k$, then \Main\
	returns a value $\hnk$ such that with probability at least $1-\delta$, $\hnk< \tnk$.
	\item The expected running time and query complexity of the algorithm are
		\[O\left(\min\left\{ \frac{n \alpha^{k-1}}{\tnk},\; \frac{n}{n_k^{1/k}} \right\} +
	\frac{m \cdot \alpha^{k-2}}{\tnk}\cdot	\frac{n_k}{\tnk}\right)\cdot
	\poly(\log (n/\delta), 1/\epsilon, k^k).
	\]
	\end{itemize}
\end{lemma}
\begin{proof}
	By~\cite{ERS17}, each invocation of their algorithm with parameters $n,\alpha$ and $\eps=1/2$, returns a factor-$2$ approximation of $m$, with probability at least $2/3$.
	Hence, with probability at least $1-\delta/2$, $\tm \in [m/2,2m]$.
	Condition on this event.

	Consider (as a mental experiment) invoking \IsActive$(i,\vI,k,\alpha,\eps,\delta/4,\tnk,\tm, \vec{\tau})$
     on all ordered cliques $\vI$ of size $i\in [k-1]$
	with $\tm$ as computed above, $\vec{\tau}$ set to $\vtauu$ (and $\alpha$, $\eps$, and $\tnk$ as provided to \Main).
	Let	$\mA$ be the subset of ordered cliques on which \IsActive\ returns {\sf Active}.
	By Lemma~\ref{lem:IsActive}, if $\tnk\in [n_k/4, n_k]$, then with probability at least $1-\delta/4$,
 $\mA$ is $(\eps,\vtauu)$-good and if $\tnk > n_k$, then with probability at least $1-\delta/4$,
    $\mA$ is 	$\vtauu$-bounded. Condition on this event as well.
	
		If $\tnk\in [n_k/4, n_k]$, then by Theorem~\ref{thm:approx-weight}, each invocation of \Approx\ returns a value that is
 in the interval $(1\pm\eps)n_k$ with probability at least $5/6$.
	It follows that  with probability at least $1-\delta/4$, the median of the returned values, $\hnk$, is in $(1\pm\eps)n_k$.
 The first item of the lemma follows by taking a union bound over the event that $\tm \notin [m/2,2m]$,
 the event that $\mA$ is not $(\eps,\vtauu)$-good (conditioned on $\tm \geq m/2$), and the event that
$\hnk \notin (1\pm\eps)n_k$ (conditioned on $\tm \geq m/2$ and $\mA$ being $(\eps,\vtauu)$-good).
The second item in the lemma follows similarly from the case $\tnk > n_k$ in Theorem~\ref{thm:approx-weight}.
	
	
	It remains to analyze the complexity of \Main.
	By Theorem~\ref{thm:approx-weight}, the running time of \Approx\ when invoked with parameters $n,k,\alpha,\eps, \delta'=1/6,\tm,\tnk,\vec{\tau}$ and an oracle $\op$ is
	\[O\left(\frac{n \tau_1}{\tnk} + \frac{k^{3k}}{(\delta')^{2k}}\cdot \frac{\tm}{\tnk}\cdot
	\sum_{t=2}^{k-1} (\alpha^{t-2}\cdot \tau_t)
	+  \frac{k^{3k}}{(\delta')^{k}}\cdot   \frac{n_k}{\tnk}\cdot  \tau_k \cdot \rho(\op)
	\right)\cdot \frac{3^k\cdot \ln(k/\delta')}{\eps^2}\;, \numberthis \label{eqn:Approx}
	\]
	and the upper bound on the query complexity of \Approx\ is obtained by exchanging $\rho(\op)$ with $q(\op)$.
	
	Recall that we replace each call to the oracle $\op$ on an ordered $i$-clique $\vI$ with an invocation of the procedure \IsActive\ with parameters $i, \vI, k,\alpha, \delta/4,\tnk, \tm$ and $\vec{\tau}=\vtauu$. By applying  Lemma~\ref{lem:IsActive} we get that
	\[q(\op)\leq \rho(\op)=O\left(\frac{\tm}{\tnk}\cdot \sum_{t=2}^{k} (\alpha^{t-2}\cdot \tauu_t) \right)\cdot \frac{(6k)^{6k} \cdot  \log^2(n/\delta)}{\eps^3}\;.
	\numberthis\label{eqn:oracle}
	\]
	

	It follows from Equation~\eqref{eqn:Approx}, Equation~\eqref{eqn:oracle} and the
	 setting of $\vtauu$ in Definition~\ref{def:tau_i},
	that the query complexity and running time resulting from all invocations of \Approx\ are
	\[O\left(\min\left\{ \frac{n \alpha^{k-1}}{\tnk},\; \frac{n}{n_k^{1/k}} \right\} +
	\frac{\tm \cdot \alpha^{k-2}}{\tnk}\cdot
	\frac{n_k}{\tnk}  	\right)\cdot
	\poly(\log (n/\delta), 1/\epsilon, k^k). \numberthis
	\]
	
By~\cite{ERS17}, the expected query complexity and running time of each invocation of their algorithm is $O\left(\frac{n \alpha}{m}\right)\cdot (\log n\cdot \log \log n/\eps^3)$, which is negligible compared to the running time of \Approx. Furthermore, since $\tm > 2m$ with probability at most $1/n^2$, $\tm \leq n^2$ (or otherwise we will set $\tm=n^2$), and since the running time of \Approx\ grows linearly with $\tm$, this event does not affect the expected query complexity and running time of \Main.
	Therefore, the expected query complexity and running time of \Main\ are
	\[O\left(\min\left\{ \frac{n \alpha^{k-1}}{\tnk},\; \frac{n}{n_k^{1/k}} \right\} +
	\frac{m \cdot \alpha^{k-2}}{\tnk}\cdot
	\frac{n_k}{\tnk}\cdot
	\right)\cdot
	\poly(\log (n/\delta), 1/\epsilon, k^k)\;, \numberthis
	\]
	as claimed.
\end{proof}

It remains to alleviate the need for a (coarse) estimate $\tnk$ of $n_k$, and obtain the next corollary that restated Theorem~\ref{thm:main}.

\newcommand{\ov}{\overline{v}}
\begin{corollary}\label{cor:main}
There exists an algorithm that,  given query access to a graph $G$
and parameters $n,\alpha$ and  $\eps$, returns a value  $\hnk$ such that with probability at least $2/3$, $\hnk \in (1\pm \eps)$. The expected query complexity is running time of the algorithm are
\[O\left(\min\left\{ \frac{n \alpha^{k-1}}{n_k},\; \frac{n}{n_k^{1/k}} \right\} +
\min\left\{ \frac{m \cdot \alpha^{k-2}}{n_k}, m\right\}  	\right)\cdot
\poly(\log (n), 1/\epsilon, k^k).
\]
and the expected running time is
\[O\left(\min\left\{ \frac{n \alpha^{k-1}}{n_k},\; \frac{n}{n_k^{1/k}}  + \frac{m \cdot \alpha^{k-2}}{n_k} \right\}
 	\right)\cdot
\poly(\log (n), 1/\epsilon, k^k).
\]
\end{corollary}
\begin{proof}
\sloppy
		We prove the corollary by relying on Theorem~18 of~\cite{ERS18}, which we refer to as the {\em Search Theorem\/}.
		Let $\vec{V}=(n,k,\alpha)$ and let $\mAlg(\tnk, \eps,\delta, \vec{V})=$\Main$(n,k,\alpha, \eps, \delta,\tnk)$.
		By Lemma~\ref{lem:Main}, an invocation of $\mAlg$ with parameters
$\tnk, \eps$, $\vec{V}=(n,k,\alpha)$ and $\delta=\eps/\log \log (n^k)$,
		 meets (with a small caveat that will be discusses momentarily) the first two requirements of the search theorem.
		  Let $\mu\big(\mA(\tnk, \eps,\delta, \vec{V})\big)$ denote the expected 
running time of $\mAlg(\tnk, \eps,\delta, \vec{V})$. It holds that $\mu\big(\mAlg(\tnk, \eps,\delta, \vec{V})\big)$ is monotonically non-increasing with $\tnk$, and that  $\mu\big(\mAlg(\tnk, \eps,\delta, \vec{V})\big) \leq  \mu\big(\mAlg(\tnk, \eps,\delta, \vec{V})\big) \cdot (\tnk/n_k)^{\ell}$ for $\ell=2$. Also, the maximal value of $n_k$ is $B=n^k$.

\sloppy		
	Therefore, by the Search Theorem, there exists an algorithm that, given access to \Main, returns a value $\hnk$ such that $\hnk \in (1\pm \eps)$ with probability at least $2/3$.
	Furthermore, the expected running time of the algorithm is $\mu\big(\mbox{\Main}(n, k,\alpha, \eps, \delta, \tnk)\big)\cdot \poly(k \log n, 1/\epsilon)$. By Lemma~\ref{lem:Main}, this equals
	\[O\left(\min\left\{ \frac{n \alpha^{k-1}}{\tnk},\; \frac{n}{n_k^{1/k}} \right\} +
	\frac{m \cdot \alpha^{k-2}}{n_k} 	\right)\cdot
	\poly(\log (n/\delta), 1/\epsilon, k^k).
	\numberthis \label{eqn:complex}
	\]
		
	The aforementioned caveat is that the second item in Theorem~18 of~\cite{ERS18} requires that if $\tnk>n_k$, then with probability at least $\eps/4$, \Approx\ returns a value such that $\hnk<(1+\eps)n_k$, while the discussion at the beginning of the proof only gives that if $\tnk>n_k$ then with probability at least $1-\delta$, $\hnk \leq \tnk$. However, it can be easily verified that this condition is also sufficient for the proof of Theorem 18 to hold\footnote{This inequality is only used in the last equation in the first column of page 732, which holds if the second term is removed.}.
	
	Finally, since each pair query of the algorithm is preceded by a neighbor query, by saving the answers to its previous queries, the algorithm can avoid performing more than $O(m)$ queries. Hence, the corollary follows from Equation~\eqref{eqn:complex}.
\end{proof}

\newcommand{\INT}{\textsf{INT}^N}
\newcommand{\calE}{\mathcal{E}}
\newcommand{\ALG}{ALG}
\section{Lower Bound}\label{sec:lb}

In~\cite{Blais2012}, Blais et al. developed a framework for proving property testing lower bounds via reductions from communication complexity. Their work was later generalized by Goldreich~\cite{Gol13_CC} and recently was formalized for the setting of sublinear graph estimation
by Eden and Rosenbaum~\cite{ER_CC}. All of our constructions in this section are slight variations of the constructions described in~\cite{ER_CC} (which themselves rely on previous papers).

We start with a very high-level overview of the framework of~\cite{ER_CC} for proving lower bounds on parameter estimation problems.
Let $P$ be the parameter at question, let \textsf{Approx-P} be an algorithm for approximating $P$ up to some factor and with high probability, and let $Q$ be the set of \textsf{Approx-P}'s allowable queries.
 The starting point is  choosing a ``hard'' communication problem $\Pi:\{0,1\}^N\times\{0,1\}^N \rightarrow \{0,1\}$ and reducing solving $\Pi$ to solving the estimation problem of $P$.  	This is done by defining an embedding $\mathcal{E}: \{0,1\}^N\times\{0,1\}^N \rightarrow \mG_n$, where $\mG_n$ is the family of graphs over $n$ vertices, as follows.
 \[
P(x,y)=\begin{cases}
p & \Pi(x,y)=0\\
p' & \Pi(x,y)=1
\end{cases}
 \]
 such that an algorithm for estimating $P$ 
 can differentiate between $p$ and $p'$. Furthermore, the embedding function $\calE$ should be such that for any input $x,y$, Alice and Bob can answer any query form $Q$ on $\calE(x,y)$ by at most $O(\beta)$ bits of communication. 
 Thus, Alice and Bob can solve $\Pi(x,y)$ by invoking \textsf{Approx-P} on the (implicit) graph $\calE(x,y)$, answering its queries by communicating and at last answering according to \textsf{Approx-P}'s result. This implies that $QC(\mbox{\textsf{Approx-P}})=\Omega(CC(\Pi) /\beta)$, where $QC(\mbox{\textsf{Approx-P}})$ is the expected query complexity of \textsf{Approx-P} and $CC(\Pi)$ is the lower bound on the expected query complexity of $\Pi$.

Here we consider the following (hard) communication problem.
Let $\INT_r$ be a promise problem that receives two inputs $x,y:\{0,1\}^N \times \{0,1\}^N$ such that either $\sum_{i=1}^N x_i \cdot y_i=r$ or $\sum_{i=1}^N x_i \cdot y_i=0$, and
\[\INT_r(x,y)=
\begin{cases}
1 & \sum_{i=1}^N x_i \cdot y_i=r \\
0 & \sum_{i=1}^N x_i \cdot y_i=0
\end{cases}
\;.
\]
This problem is a simple generalization of the well-known communication problem Set-Disjointness, denoted $\textsf{Set-Disj}$, in which Alice and Bob are required to distinguish between the case that the inputs $x$ and $y$ intersect in a single bit to the case that they do not intersect at all. The randomized communication complexity of $\textsf{Set-Disj}$ is $\Omega(N)$ even for the case that Alice and Bob have access to shared randomness and are only required to reply correctly with high constant probability~\cite{Kalyanasundaram1992,Razborov1992}. An easy corollary is that $CC(\INT_r) = \Omega(CC({\textsf{Set-Disj}})/r)=\Omega(N/r)$ (see Corollary 2.7 in \cite{ER_CC}).

\begin{thm}\label{thm:lb-full}
Let $Q$ be the set of degree, neighbor and pair queries.
Let $\ALG$ be a multiplicative approximation algorithm for estimating the number of $k$-cliques that succeeds with probability at least $2/3$ for any graph $G$ with $n$ vertices, $m$ edges, $n_k$ cliques and arboricity at most $\alpha$ and let $q(\ALG)$ denote its expected query complexity over the set of queries $Q$. It holds that  \[q(\ALG)=\Omega\left(\min\left\{\frac{n}{k^k \cdot n_k^{1/k}}, \frac{n \alpha^{k-1}}{k\cdot n_k}\right\}+\min\left\{\frac{m(\alpha/k)^{k-2}}{n_k},m\right\} \right)\;.\]
\end{thm}


\begin{proof} 		
For the sake of simplicity, for parameters $n,m$ and $n_k$,
we shall 
consider graphs
with $\Theta(n)$ vertices, $\Theta(m)$ edges, and $\Theta(n_k)$ $k$-cliques.
We shall prove that for different settings of these parameters, it is difficult to distinguish
(with fewer queries than stated in the theorem) between such graphs that have arboricity at most $\alpha$ and graphs over the same number of vertices,
the same bound on the arboricity,
and {\em no $k$-cliques\/}. (We note that the latter graphs will also have $\Theta(m)$ edges, which implies that the lower bound holds also when the algorithm is provided with a constant factor estimate of the number of edges.)
Recall that by Corollary~\ref{cor:nk_vs_nt}, for graphs with arboricity at most $\alpha$, we have that $m=O(n\alpha)$ and $n_k=O(m\alpha^{k-2}/k!)$.

\paragraph{Establishing the first (additive) term.}

We first consider the case that $n_k \leq \binom{\alpha}{k}$, and prove that in this case, $q(ALG)=\Omega\left(\frac{n}{k \cdot n_k^{1/k}}\right)$. We establish this (relatively simple) claim without relying on the communication complexity framework.
Let $w$ be the maximal integer for which $\binom{w}{k}\leq n_k$, so that $w \leq \alpha$, and a clique over $w$ vertices contains
$\Theta(n_k)$ $k$-cliques.
Observe that $(k/e)\cdot n_k^{1/k} \leq w \leq k \cdot n_k^{1/k}$.
Let $G'$ be a graph over $2n$ vertices, $m$ edges, arboricity $\alpha$ and no $k$-cliques
 (e.g., a bipartite graph with $m/\alpha$ vertices on each side, each having $\alpha$ neighbors, and the remaining $2n-2m/\alpha$ vertices are isolated vertices).
Consider the following two families of graphs, $\mG_1$ and $\mG_2$. Both families consist of the subgraph $G'$ and an additional subgraph over $w$ vertices, denoted $H$.
In all the graphs of the family $\mG_1$, $H$ is a clique, and in the graphs of the family $\mG_2$,
$H$ is an independent set.
The graphs within the families differ only by the labeling of the vertices. By the above construction, for any graph $G\in \mG_1\cup \mG_2$, $n(G)=n+w=\Theta(n)$, as $w\leq k \cdot n_k^{1/k}\leq n$ where the last inequality is by Corollary~\ref{cor:nk_vs_nt}. Since $w\leq \alpha$,
in both families graphs have $\Theta(m)$ edges (since $\binom{w}{2}\leq \binom{\alpha}{2}\leq m$),
and the graphs in
both families have arboricity at most $\alpha$.
Finally, for any $G\in \mG_1$, $n_k(G)=\Theta(n_k)$, and for any $G\in \mG_1$, $n_k(G)=0$.

Clearly, in order for any algorithm to distinguish with high constant probability between graphs drawn from the first family, and graph drawn from the second family, it must hit the set $H$ with high constant probability. Since the probability of hitting a vertex in the set $H$ is $w/n$, it follows that distinguishing between the two families requires
$\Omega(n/w) = \Omega\left(\frac{n}{k \cdot n_k^{1/k}}\right)$ queries in expectation.


We now 
show that for $n_k>\binom{\alpha}{k}$,
$q(\ALG)=\Omega\left(\frac{n\alpha^{k-1}}{k^k\cdot n_k}\right)$.
Let $r$ be the maximal integer such that $r\binom{\alpha}{k} \leq n_k$.
We will reduce from the aforementioned problem $\INT_r$  for $N= n/\alpha $, and $r$ as above.
(We assume for the sake of simplicity that $n$ is divisible by $\alpha$ - the construction can be easily modified if this is not the case.)
For two inputs $x,y\in \{0,1\}^N$ we define $G_{x,y}$ as follows.
For every index $i \in [N]$, there is a set $S_i$ over $\alpha$ vertices that is either an independent set if $x_i \cdot y_i =0$ or a clique if $x_i \cdot y_i =1$. In addition, for every $x$ and $y$ the graph $G_{x,y}$ contains a fixed subgraph $G'$ over $n$ vertices  and $m$ edges, arboricity $\alpha$ and no $k$-cliques.
Hence, if $\sum _{i=1}^N x_i \cdot y_i=r$ then $n_k(G_{x,y})=r\cdot \binom{\alpha}{k} = \Theta(n_k)$, and if $\sum _{i=1}^N x_i \cdot y_i=0$ then $n_k(G_{x,y})=0.$ Also, by Corollary~\ref{cor:nk_vs_nt}, $n_k \leq n\alpha^{k-1}/k!$, so for any $x,y$ $n(G_{x,y})=n + r\cdot \alpha = \Theta(n)$. For $x,y$ such that $\sum _{i=1}^N x_i \cdot y_i=r$ it holds by Corollary~\ref{cor:nk_vs_nt}, that $m(G_{x,y})=m + r\cdot \binom{\alpha}{2} = \Theta(m)$, and for $x,y$ such that $\sum _{i=1}^N x_i \cdot y_i=0$, $m(G_{x,y})=m$. Finally, for any $x,y$, $\alpha(G_{x,y})=\alpha$.

 Alice and Bob can answer any degree, neighbor or pair query by exchanging at most two bits. Specifically, if the query is a degree query about a vertex in the $S_i$, then its degree can be determined by Alice and Bob exchanging the value of their $i\th$ bit. Similarly, if the query is a neighbor query on a vertex in $S_i$ or a pair query on two vertices in a set $S_i$. On all other queries Alice and Bob can answer with no communication.
  Therefore, by Theorem 3.3 of~\cite{ER_CC},  $q(\ALG)=\Omega(N/r)/2=\Omega\left(\frac{n\alpha^{k-1}}{k^k\cdot n_k}\right)$.

\paragraph{Establishing the second term.}
	We first deal with the case that $n_k=\Theta(r\cdot (\alpha/k)^{k-2})$ for some integer $r\geq 1$, and prove that $q(\ALG)=\Omega\left(\frac{m(\alpha/k)^{k-2}}{n_k}\right)$.
	We use a variation of the proof of Theorem B.1 of~\cite{ER_CC}.
	We reduce from the problem  $\INT_r$ for $N=m$ and $r$ as above, where we  view strings in $\{0,1\}^m$ as strings in $\{0,1\}^{m/\alpha}\times \{0,1\}^{\alpha}$. That is, we view the indices of $x$ as pairs $(i,j)$ for $i \in [m/\alpha]$ and $j \in [{\alpha}]$. 
	For every $x,y$  the graph $G_{x,y}$ consists of sets $A,A',B,B'$,  sets $S_1, \ldots, S_{k-2}$ and an additional set $C$ as detailed next.  The size of each of the sets $A,A',B,B'$  is $m/\alpha$, the size of each of the sets $S_1, \ldots, S_{k-2}$ is $\alpha/k$ and the size of $C$ is $n-m/\alpha$. (Here we assume for the sake of simplicity that $m$ is divisible by $\alpha$ and $\alpha$ is divisible by $k$.)
There is an edge between every vertex in   $A\cup B$  and every and vertex in $S=S_1\cup \ldots\cup S_{k-2}$. Also,
	there is a complete bipartite graph between any two sets $S_{\ell}$ and $S_{\ell'}$ for every $\ell, \ell'\in[k-2]$. The set $C$ is an independent set.
	The rest of the edges of the graph $G_{x,y}$ depend on $x$ and $y$ as follows. Let $A=\{a_1, \ldots, a_{m/\alpha}\}$, $B=\{b_1, \ldots, b_{m/\alpha}\}$ and similarly for $A'$ and $B'$.
	For the sake of simplicity, when we write $a_{i}$ we mean $a_{(i \mod n + 1)}$, and similarly for all other vertices in $A',B$ and $B'$.
	For any two indices $i\in [n], j \in [{\alpha}]$, if $x_{i,j} \cdot y_{i,j}=0$ the graph $G_{x,y}$ contains the edges $(a_i,a'_{i+j})$ and $(b_{i+j},b'_{i})$ and for any two indices $i\in[n], j \in [\alpha]$ such that $x_{i,j}\cdot y_{i,j}=1$ the graph $G_{x,y}$ contains the edges $(a_i,b_{i+j})$ and $(a'_{i+j}, b'_i)$.
	Furthermore, for every $a_i$, we label the neighbors to $A'\cup B$ as the first $\alpha$ neighbors, and the remaining $(k-2)\cdot \alpha/k$ neighbors are the vertices of $S$ (with some arbitrary but fixed order). We similarly label the neighbors of $B$.
	
	It follows that if $\INT_r(x,y)=0$ then $n_k(G_{x,y})=0$ and otherwise $n_k(G_{x,y})=r \cdot (\alpha/k)^{k-2}$ (as every edge between the sets $A$ and $B$ creates $(\alpha/k)^{k-2}$ $k$-cliques with the vertices of the sets $S_1, \ldots, S_{k-2}$).
	Also, for any $x$ and $y$,
	$m(G_{x,y})=2m+(m/\alpha)\cdot(k-2)\cdot (\alpha/k)+ \binom{k-2}{2}(\alpha/k)^2=\Theta(m)$, where the last inequality is due to Corollary~\ref{cor:nk_vs_nt}. Finally, for every $x$ and $y$, $\alpha(G_{x,y})=\alpha.$
	
	It remains to prove that Alice and Bob can answer any degree, neighbor or pair query with bounded communication.	
	 Observe that the degrees of the vertices in the graph are not affected by the input $x,y$. Therefore, Alice and Bob can answer degree queries with zero communication. This is also the case for neighbor queries on the vertices of the sets $S_1, \ldots, S_{k-2}$ and of $C$, as well as neighbor queries $(v_i,j)$ for $v\in\{a,a',b,b'\}$ and $j>\alpha$ (by the order of labels defined above).
	 For any neighbor query $(v_i,j)$ such that $v \in \{a,a',b,b'\}$ and $j\in [\alpha]$, Alice and Bob can respond by communicating $x_{i,j}$ and $y_{i,j}$ to each other and deciding according to $x_{i,j}\cdot y_{i,j}$.
	For example, if the algorithm performs a neighbor query $(a_i, j)$ then Alice sends Bob $x_{i,j}$ and Bob sends Alice $y_{i,j}$ and if $x_{i,j}\cdot y_{i,j}=0$ then they respond $a'_{i+j}$ and otherwise $b_{i+j}$.	
	Pair queries $(v_i,u_{i+j})$ for $v_i\in A$ and $u_{i+j} \in A' \cup B$
	can again be answered by exchanging the two bits $x_{i,j}$ and $y_{i,j}$ and similarly pair queries for pairs $(v_i,u_{i+j})$ for $v_i\in B$ and $u_{i+j} \in B' \cup A$. All other pair queries can be answered without any communication. Hence, all queries can be answered by exchanging at most $\beta=2$ bits of communication.
	Therefore, by Theorem 3 of~\cite{ER_CC},
	$q(\ALG)= \Omega((N/r)/\beta)=\Omega(m/r)=\Omega\left(\frac{m(\alpha/k)^{k-2}}{n_k}\right).
	$

	In the case where $n_k\leq (\alpha/k)^{k-2}$, we modify the above construction as follows. We reduce the size of each set in $S$ to $n_k^{\frac{1}{k-2}}$, and we let $r=1$. Hence, if $\INT_r (x,y)=0$ then $n_k(\calE(x,y)) = 0$, and if $\INT_r (x,y)=1$ then $n_k(\calE(x,y)) = \Theta(n_k)$.
	The proof follows as before, and we get that
	$q(\ALG)= \Omega((m/r)/\beta) =\Omega(m)$.
\end{proof}

\ifnum\conf=0
\bibliographystyle{alpha}
\else
\bibliographystyle{plain}
\fi
\bibliography{k_cliques_bib}			
 \newpage
\appendix

\section{Table of notations}\label{app:prel}

In the following table we gather the various notations used throughout the paper. 
In all notations, $\vT$ is an ordered $t$-clique $(v_1,\dots,v_t)$ (for some $t\in[k]$) 
and $\mR$ is a (multi)-set of ordered $t$-cliques. Each notation also has a hyperlink to the location where it was first defined.

\hypersetup{linkcolor=black}

\begin{centering}
	\begin{table}[ht!]
 \begin{tabular}{ !{\vrule width 1.3pt} m{5.4cm} | m{10cm} !{\vrule width 1.3pt}}
 	     \noalign{\hrule height 2pt}
			{\bf Notation} & {\bf Meaning} \\ 	
    \noalign{\hrule height 2pt}
   \vspace{1ex}
		\hyperref[par:notation]{$U(\vT)$}  & The unordered clique $\{v_1, \ldots, v_t\}$ corresponding to $\vT=(v_1, \ldots, v_t)$. \\  \hline
         \vspace{1ex}
		\hyperref[par:notation]{$\mC_t = \mC_t(G)$} &  The set of $t$-cliques in $G$. \\ \hline
         \vspace{1ex}
		\hyperref[par:notation]{$n_t=|\mC_t|$} & The number of $t$-cliques in $G$. \\  \hline
         \vspace{1ex}
        \hyperref[def:least_deg_ver]{$d(\vT)$}  &  The degree of the minimal-degree vertex in $\vT$ \\ \hline
         \vspace{1ex}
        \hyperref[def:least_deg_ver]{$d(\mR) =\sum_{\vT\in \mR} d(\vT)$}  &  The sum of degrees of minimal-degree vertices in ordered cliques belonging to $\mR$ 
                                  \\ \hline
         \vspace{0.7ex}
		\hyperref[par:Tleqj]{$\vT_{\leq j}$} for $j \leq t$ & The $j$-tuple $(v_1,\dots,v_j)$ formed by the first $j$ elements in $\vT$. \\ \hline
        \vspace{1ex}
		\hyperref[par:tuple]{$(\vT, u)$} & The ordered $(t+1)$-tuple $(v_1, \ldots, v_t, u)$. \\ \hline
         \vspace{1ex}
		\hyperref[par:mC_jT]{$\mC_j(\vT)= \mC_j(U(\vT))$} for $j\geq t$ & The set of unordered $j$-cliques that $\vT$ participates in. That is, the set of  $j$-cliques $T'$ such that $U(\vT) \subseteq T'$.\\ \hline
        \vspace{1ex}
		\hyperref[par:mC_jT]{$c_j(\vT)=|\mC_j(\vT)|$}, for $j\geq t$ & The number of $j$-cliques that $\vT$ participates in.\\ \hline
         \vspace{1ex}
		\hyperref[par:mC_jT]{$\mC_j(\mR)$}  for $j \geq t$ & The union of $\mC_j(T)$ taken over
		all $\vT \in \mR$, where here in ``union'' we mean with multiplicity.
             \\ \hline
              \vspace{1ex}
				\hyperref[par:c_k(T)]{$c_k(\mR) =\sum_{\vT\in \mR}c_k(\vT) %$ \phantom{ck(R)x}  $
                 =|\mC_k(\mR)|$} &  The number of $k$-cliques that ordered cliques in $\mR$ participate in (with multiplicity) \\ \hline
                  \vspace{1ex}
		\hyperref[par:mC_jT]{$\mOC_j(\vT)$} for $j \geq t$ & The set of all ordered $j$-cliques  $\vT' = (v_1,\dots,v_t,u_{t+1},\dots,u_j)$ that are extensions of $\vT$. \\ \hline
   \vspace{1ex}
	 	\hyperref[par:mC_jT]{$\mOC_j(\mR)$} for $j \geq t$ & The union of $\mOC_j(T)$ taken over
	 	all $\vT \in \mR$ (with multiplicity). \\ \hline
  \vspace{1ex}
  \hyperref[par:active]{$\mA$} & A subset of active ordered cliques.\\ \hline
  \vspace{1ex}
  \hyperref[def:fully-active]{$\mFA_t$}  & The set of fully-active ordered $t$ cliques w.r.t. $\mA$, i.e., they and each of their prefixes belong to $\mA$.\\ \hline
  \vspace{1ex}
  \hyperref[par:wt]{$\wtA(\vT)$} & The  weight of $\vT$ with respect to $\mA$. \\ 
  \hline
  \vspace{1ex}
  \hyperref[def:fully-active-I]{$\mFAI_t$}  & The set of fully-active ordered $t$ cliques w.r.t. $\mA$ and $\vI$.\\ \hline
  \vspace{1ex}
  \hyperref[par:wtI]{$\wtAI(\vT)$} & The  weight of $\vT$ with respect to $\mA$ and $\vI$. \\ 
  
 	    \noalign{\hrule height 1.5pt}
	\end{tabular}
\caption{The notations used throughout the paper.} \label{tab:notation}
\end{table}
\end{centering}
We note that we use $\mC_j(\cdot)$ in two different ways: $\mC_j(G)$ is the number of $j$-cliques {\em in\/} the graph $G$, while $\mC_j(T)$ is the number of $j$-cliques (in $G$) that $T$ {\em participates in\/} (belongs to). Since we use $\mC_j$ as a shorthand for $\mC_j(G)$, there should			be no ambiguity.

\hypersetup{linkcolor=blue} 

\ifnum\conf=1
\section{Missing proofs from Section~\ref{sec:prel}}\label{appendix:prel-proofs}

\SumMindCliques

\medskip

\BoundNkProof

\fi

\end{document}